\documentclass[a4paper,UKenglish]{lipics-v2016}

\usepackage{microtype}%
\usepackage{apxproof}

\usepackage{amsmath,amssymb}
\usepackage{xspace}
\usepackage{todonotes}
\usepackage{algorithm}
\usepackage[noend]{algpseudocode}
\usepackage[inline,shortlabels]{enumitem}

\theoremstyle{plain}
\newtheorem{observation}[theorem]{Observation}

\newcommand{\calA}{\mathcal{A}}
\newcommand{\calB}{\mathcal{B}}
\newcommand{\calC}{\mathcal{C}}

\newcommand{\calI}{\mathcal{I}}
\newcommand{\calN}{\mathcal{N}}
\newcommand{\calQ}{\mathcal{Q}}
\newcommand{\calP}{\mathcal{P}}

\newcommand{\calS}{\mathcal{S}}
\newcommand{\calT}{\mathcal{T}}
\newcommand{\calV}{\mathcal{V}}
\newcommand{\calX}{\mathcal{X}}

\newcommand{\intcomp}[1]{\calI\calC_{#1}}

\newcommand{\eval}[1]{{\sc EVAL(#1)}}
\newcommand{\cqeval}[1]{{\sc CQ-EVAL(#1)}}
\newcommand{\pcqeval}[1]{\textup{p}-{\sc CQ-EVAL(#1)}}
\newcommand{\cqevalp}{{\sc CQ-EVAL}}

\newcommand{\peval}[1]{\textup{p}-{\sc Eval(#1)}}

\newcommand{\ext}[1]{{\sc EXT(#1)}}
\newcommand{\extp}{\mathrm{EXT}}
\newcommand{\pExt}[1]{\textup{p}-{\sc EXT(#1)}}
\newcommand{\pHom}{\mathrm{p\textup{-}Hom}}
\newcommand{\Hom}{\mathrm{HOM}}

\newcommand{\branch}{\mathsf{branch}}
\newcommand{\cbranch}{\mathsf{cbranch}}

\newcommand{\var}{\mathsf{var}}
\newcommand{\ievar}{\mathsf{evar}}
\newcommand{\ifvar}{\mathsf{fvar}}
\newcommand{\idom}{\mathsf{dom}}

\newcommand{\ptime}{\mathsf{PTIME}}
\newcommand{\PSPACE}{\mathsf{PSPACE}}
\newcommand{\FPT}{\mathsf{FPT}}
\newcommand{\NP}{\mathsf{NP}}
\newcommand{\coNP}{\mathsf{coNP}}
\newcommand{\SigmaP}[1]{\Sigma_{#1}\mathsf{P}}
\newcommand{\wone}{\mathsf{W}[1]}
\newcommand{\cowone}{\mathsf{coW}[1]}
\newcommand{\icore}{\mathsf{core}}
\newcommand{\extcore}{\mathsf{extcore}}
\newcommand{\anspred}{\mathit{Ans}}

\newcommand{\str}[1]{\mathbf{#1}}
\newcommand{\strD}{\str D}

\newcommand{\pt}{PT\xspace}
\newcommand{\pts}{PTs\xspace}
\newcommand{\wdpt}{wdPT\xspace}
\newcommand{\wwdpt}{wwdPT\xspace}
\newcommand{\wdpts}{wdPTs\xspace}

\newcommand{\cq}{CQ\xspace}
\newcommand{\cqs}{CQs\xspace}

\newcommand{\rcs}{\mathsf{rcs}}

\newcommand{\ra}{\rightarrow}

\newcommand{\tractbox}[2]{
	\noindent
	\fbox{\begin{minipage}{\textwidth-6pt}\textbf{Tractability condition (#1):}#2
		\end{minipage}}	
}

\newcommand{\csts}{\mathsf{csts}}
\newcommand{\CST}{\mathsf{CST}}
\newcommand{\extcq}{\mathsf{extcq}}
\newcommand{\cia}[1]{\mathsf{cia}(#1)}

\newcommand{\relvnode}[1]{\mathsf{relv}(#1)}
\newcommand{\stopmaps}{\mathit{stop}}
\newcommand{\extmaps}{\mathit{extend}}

\theoremstyle{plain}
\newtheoremrep{proposition}[theorem]{Proposition}
\newtheoremrep{lemma}[theorem]{Lemma}

\title{On tractable query evaluation for SPARQL}
\titlerunning{On tractable query evaluation for SPARQL} %

\author[1]{Stefan Mengel}
\author[2]{Sebastian Skritek}
\affil[1]{CNRS, CRIL UMR 8188, Lens, France;
  \texttt{mengel@cril.fr}}
\affil[2]{Faculty of Informatics, TU Wien;
  \texttt{skritek@dbai.tuwien.ac.at}}
\authorrunning{S.\ Mengel and S.\ Skritek} %

\Copyright{Stefan Mengel and Sebastian Skritek}%

\EventEditors{}
\EventNoEds{0}
\EventLongTitle{}
\EventShortTitle{}
\EventAcronym{}
\EventYear{}
\EventDate{}
\EventLocation{}
\EventLogo{}
\SeriesVolume{}
\ArticleNo{00}

\begin{document}

\maketitle

\begin{abstract}
Despite much work within the last decade on foundational properties of
SPARQL -- the standard query language for RDF data -- rather little is
known about the exact limits of tractability for this language. In
particular, this is the case for 
SPARQL
queries that contain the OPTIONAL-operator, even though 
it is %
one of the most intensively studied features of SPARQL.
The aim of our work is to provide a more thorough picture of tractable
classes of SPARQL queries.

In general, SPARQL query evaluation is $\PSPACE$-complete in combined complexity, and it remains $\PSPACE$-hard
already for queries containing only the OPTIONAL-operator. To amend this situation, research has focused on 
``well-designed SPARQL queries'' and their recent generalization ``weakly well-designed SPARQL queries''. For these two fragments 
the evaluation problem is $\coNP$-complete in the absence of projection and 
$\SigmaP{2}$-complete otherwise. 
Moreover, they have been shown to contain most SPARQL queries asked in practical settings.

In this paper, we study 
tractable classes of weakly well-designed queries 
in parameterized complexity considering the equivalent formulation as pattern trees. 
We give a complete characterization of the tractable classes in the case without projection. 
Moreover, we show a characterization of all tractable classes of \emph{simple} 
well-designed pattern trees in the presence of projection. 
\end{abstract}

\section{Introduction}
Driven by the increasing amount of RDF data available on the web, SPARQL---the 
standard query language for RDF---has received a lot of attention
in the last decade. Many different aspects of SPARQL have been studied,
for example its expressive power \cite{Pol2007,AG2008,PW2013,ZB2014},
complexity of query evaluation and optimization 
\cite{PAG2009,SML2010,PV2011,CEGL2012,CEGL2012b,LM2013,LPPS2013,PS2014,KRRV2015,BPS2015,KPS2016,CP2016}, 
semantic properties \cite{AP2011,AFPSS2015,AU2016}, and several more 
\cite{AG2016b,AG2016,AACP2013,KRU2015,ADK2016,KKG2017,GUKFC2016,ZBP2016}. 
One of the main features of SPARQL that attracted a lot of interest is the
{\sf OPTIONAL} operator that %
resembles the left outer join in the
Relational Algebra.
It allows users to define parts in a query for which an answer
is returned if possible. However, in case that providing a complete answer
including all optional parts is impossible, only a partial answer covering
those parts of the query that can be answered is returned.
Thus such partial answers are not lost.

This enables queries to retrieve meaningful answers even over incomplete information
or information provided under schemas which users do not have a good understanding of.
Given that both these characteristics---being incomplete and not well understood by
all users---are part of the nature of web data, the {\sf OPTIONAL} operator is an
essential feature of SPARQL. Thus research on SPARQL focused
on the {\sf OPTIONAL} operator \cite{PAG2009,AP2011,LPPS2013,PS2014,KK2016,PV2011,AU2016,AFPSS2015}. 

Unfortunately already early research revealed that in some cases
the semantics of the {\sf OPTIONAL}-operator can be unintuitive and inconsistent with
some principles of the semantic web \cite{AP2011}. 
To deal with this situation,
the fragment of \emph{well-designed SPARQL queries} was
introduced in \cite{PAG2009} and intensively studied later on 
\cite{LPPS2013,AP2011,PS2014,PV2011,KRU2015,BPS2015,ADK2016}.
The definition of well-designed queries forbids certain patterns of variable distributions
over {\sf OPTIONALs} which turn out to be responsible for the unintuitive semantics. Forbidding them leads
to a cleaner semantics for well-designed queries.

Regarding the evaluation of SPARQL queries, it was already shown in \cite{PAG2009} that
the problem is in general $\PSPACE$-complete in combined complexity, where the 
unintuitive behavior of the {\sf OPTIONAL}-operator
was identified as the main culprit \cite{SML2010}.
Since this behavior is absent in the well-designed fragment, as a side effect, the
complexity of the evaluation problem drops to $\coNP$-completeness for queries without 
projection \cite{PAG2009}, and $\SigmaP{2}$-completeness in case projection is allowed
\cite{BPS2015}. 

Recently, a generalization of the
well-designed queries called \emph{weakly well-designed queries}~\cite{KK2016} was proposed.
The main motivation was the observation that only
about half of the real-world queries on DBPedia are well-designed. Thus
a more relaxed condition on the variables was proposed that covers most of the real-world
queries while at the same time not increasing the complexity of query evaluation. 

Despite the wealth of research efforts on these restricted classes of queries, only little work
was done on actually identifying fragments of SPARQL containing the {\sf OPTIONAL}-operator
for which the evaluation problem is tractable. Some efforts in this direction
include \cite{BPS2015,LPPS2013}.
However, all of these results deal only with well-designed queries. Moreover, they
rely on the fact that well-designed queries can be seen as 
CQs extended by optional parts. As a consequence, their approach towards identifying
tractable fragments is to investigate to what extend tractable classes of CQs can be applied
to these queries. 
However, the exact limits of tractability have not been explored, yet.

The aim of our work is to close this gap and to provide a more thorough picture
of tractable classes of SPARQL queries containing the {\sf OPTIONAL}-operator. 
We study the complexity of query evaluation in the model of parameterized complexity
where, as usual, we take the size of the query as the parameter. As already argued in~\cite{PapadimitriouY99}, 
this model allows for a more fine-grained analysis than the classical perspectives:
on the one hand data complexity which allows impractical algorithms in which the 
size of the query is considered as a constant and on the other hand combined 
complexity where the query is assumed to have a size
similar to the database which often leads to overly pessimistic results.
In parameterized complexity, query answering is considered tractable, formally in $\FPT$, if, after a preprocessing
that only depends on the query, the actual evaluation can be done in polynomial time\cite{Grohe01,Grohe02}. 
This allows for potentially costly preprocessing on the generally small query while the dependency 
on the generally far bigger database is polynomial for an exponent independent of the query.
Parameterized complexity has found many applications in the complexity of query evaluation
problems, see e.g.~\cite{GroheSS01,Grohe2007,Marx10,Chen14}.

We remark that for Boolean Conjunctive Queries (BCQs) of bounded arity, it was shown in seminal
work of Grohe, Schwentick and Segoufin~\cite{GroheSS01} and Grohe~\cite{Grohe2007} that the tractable fragment in combined complexity 
and parameterized complexity coincide.
That is, for every class of BCQs of bounded arity the evaluation problem is in 
$\ptime$ if and only if it is in $\FPT$.
In contrast, it is known that for well-designed SPARQL queries with projection this property does not hold. This follows from~\cite{KPS2016} 
where it was shown that there are classes of well-designed queries for which the 
evaluation problem is $\NP$-hard, but fixed-paramter tractable. 
Thus the choice of the tractability notion makes a difference for the results.

To focus on the influence of the {\sf OPTIONAL}-operator, we restrict ourselves
to the \{{\sf AND},{\sf OPTIONAL}\}-fragment of SPARQL, in particular leaving out
unions and filters.
To infer our results on these queries, we will in fact work in the framework of
\emph{pattern trees} that was originally introduced in \cite{LPPS2013}
for data provided in RDF format and later extended to arbitrary relational vocabulary
\cite{BPS2015}.
Intuitively, pattern trees
represent the conjunctive parts of a query at the
nodes of the tree while the tree-structure reflects the nesting of the OPTIONALs.
Pattern trees constitute a query formalism
of their own using the ``depth-first approach'' semantics suggested in \cite{PAG2006}.
Our main technical results are characterizations of the tractable classes of pattern trees
in the setting without projection and of simple well-designed pattern trees in the presence
of projections.

From these results, one directly gets results for all fragments of 
\{{\sf AND},{\sf OPTIONAL}\} SPARQL without projection
for which the standard semantics and the depth-first semantics of \cite{PAG2006}
coincide, as is the case for the (weakly) well-designed fragment~\cite{LPPS2013,KK2016}. 
This is so because for \{{\sf AND},{\sf OPTIONAL}\} SPARQL queries one can easily compute
corresponding pattern trees by essentially just syntactic transformation.
These associated pattern trees can then be used to assess the complexity of the 
queries at hand.
Our approach thus has the advantage that, 
in case further classes of SPARQL queries for which the two possible semantics 
coincide are discovered in the future, our tractability results 
immediately carry over.

\noindent
\textbf{Summary of results and organisation of the paper.}
We study the following decision 
problem: Given a pattern tree, a database, and
a mapping, is the mapping a solution of the pattern tree over the database?
After some preliminaries in Section~\ref{sec:prelims} we will give the following results:
\begin{itemize}
	\item \emph{Tractable classes for an extension problem.}
		The semantics of weakly well-designed SPARQL queries is based on the
		idea of returning \emph{maximal mappings}. Intuitively, first the
		mandatory part of the query is mapped into the data, and then this
		partial mapping is extended as much as possible along the optional
		parts of the query. Thus, testing for extensions of partial solutions to a query
		is a central task in query evaluation. To formalize this problem, we
		introduce and analyze a problem called $\extp$ in Section~\ref{sec:extproblem}. 
		We then show that the tractable classes of $\extp$ are characterized by the
		treewidth of an auxiliary structure we call \emph{extension core}.
		This result will serve as an important building block in later sections
		and might be of interest in its own right.

	\item \emph{A complete characterization of tractable classes of pattern trees
		without projection.}
		In Section~\ref{sec:noproj}, we study the evaluation of pattern trees without projection,
		i.e., all variables occurring in the query are also part of the result. 
		Using the notions and results developed in Section~\ref{sec:extproblem}, 
		we provide a complete characterization
		of all tractable classes of both, pattern trees and weakly well-designed
		SPARQL queries.

	\item \emph{A complete characterization of tractable classes of simple 
		well-designed pattern trees with projection.}
		In Section~\ref{sec:withproj}, we study well-designed pattern trees with
		projection. 
		For technical reasons that we discuss in the conclusion, we will restrict 
		ourselves to simple pattern trees in this section, i.e.,
		pattern trees where
		no two atoms share the same relation name. This can be seen as analyzing 
		queries by their underlying ``graph structure'' similar to e.g.~\cite{GroheSS01,Chen14}
		while discarding the possibility of taking cores to simplify instances.
		Again, we provide a complete characterization of the tractable classes.
\end{itemize}
In Section Section~\ref{sec:conclusion}, we discuss our results and potential extensions to 
conclude the paper.

\section{Preliminaries}\label{sec:prelims}
{\bf Graphs.} 
We consider only undirected, simple graphs $G = (V,E)$ with 
standard notations but sometimes
write $t \in G$ to refer to a node $t \in V(G)$.
A graph $G_2$ is a subgraph of a graph $G_1$ if
$V(G_2) \subseteq V(G_1)$ and $E(G_2) \subseteq E(G_1)$.
A tree is a connected, acyclic graph. A subtree is a connected,
acyclic subgraph. A \emph{rooted tree} $T$ is a tree with 
one node $r \in T$ marked as its root. Given two nodes
$t, \hat t \in T$, we say that $\hat t$ is an ancestor of $t$ if
$\hat t$ lies on the path from $r$ to $t$.
Likewise, $\hat t$ is the parent node of $t$ ($t$ is a child of $\hat t$)
if $\hat t$ is an ancestor of $t$ and $\{t, \hat t\} \in E(T)$.
For a subtree $T'$ of $T$ that contains the root
of $T$, a node $t \in T$ is a child of $T'$ if $t \notin T'$
and $\hat t \in T'$ for the parent node 
$\hat t$ of $t$. We write $ch(T')$ for the set of all children
of $T'$.

A {\em tree decomposition} of a graph $G = (V,E)$ is a pair $(T,\nu)$,
where $T$ is a tree and $\nu : V(T) \to 2^V$, that satisfies the following:
(1) For each $u \in V$ the set $\{s \in V(T) \mid u \in \nu(s)\}$ is a
connected subset of $V(T)$, and (2) each edge of $E$ is contained in one of
the sets $\nu(s)$, for $s \in V(T)$.
The {\em width} of $(T,\nu)$ is $(\max{\{|\nu(s)| \mid s \in V(T)\}}) - 1$.
The {\em treewidth} of $G$ is the minimum width of its tree decompositions. 

\smallskip \noindent
{\bf Atoms and Conjunctive queries.} 
We assume familiarity with the relational model, especially with the
concept of {\em conjunctive queries (CQs)}, and refer to \cite{AHV}
for details. In particular, we will heavily use the fact that a conjunctive query 
can alternatively be seen as a set $\calA$ of atoms on a database $\str D$ or
as a homomorphism problem between a structure $\str A$ associated to these
atoms in a canonical way and $\str D$. In the following, we will switch between these 
perspectives whenever convenient. Sets of atoms will be denoted in calligraphic letters
$\calA, \calB, \ldots$ whereas structures will be denoted as $\str A, \str B, \ldots$.

In the following, we fix some notation.
As usual, a \emph{$\sigma$-structure} $\str A$ consists of a finite
set $A = \idom(\str A)$ %
and a relation $R^{\str A} \subseteq A^r$ for each relation symbol
$R \in \sigma$ of arity $r$.
For a set $\calA$ of atoms let 
$\var(\calA)$ denote the set of variables appearing in $\calA$.
Similarly, for a mapping $\mu$ we denote with $\idom(\mu)$ the
set of elements on which $\mu$ is defined. 
For a mapping $\mu$ and a set of variables $\calV$, we use $\mu|_{\calV}$
to describe the restriction of $\mu$ to the variables in $\idom(\mu)
\cap \calV$. 
We say that a mapping $\mu$ is an extension of a mapping $\nu$ if
$\mu|_{\idom(\nu)} = \nu$. 
By slight abuse of notation, we use operators
$\cup, \cap, \setminus$ also between sets $\calV$ and tuples $\vec v$
of variables, like in $\calV \setminus \vec v$.

A homomorphism between two
$\sigma$-structures $\str A$ and $\str B$ is a mapping $\idom(\str A)
\ra \idom(\str B)$ that, for all $R \in \sigma$ maps all tuples in
$R^{\str A}$ to tuples in $R^{\str B}$.
We write $h\colon \str A \ra \str B$ to denote a homomorphism $h$
from $\str A$ to $\str B$.
A minimal 
substructure $\str A'$ of $\str A$ such that there is a homomorphism
$\str A\rightarrow \str A'$ is called a \emph{core} of $\str A$.
We recall that all cores of
$\str A$ are unique up to isomorphism and thus speak of \emph{the} core
of $\str A$ which we denote by $\icore(\str A)$.

For a structure $\str A$ and a set $A \subseteq \idom(\str A)$, 
we write $\str A \setminus A$ to denote the restriction of
$\str A$ to $\idom(\str A) \setminus A$
(we provide a formal definition of this concepts in the appendix).
For two structures
$\str A, \str B$, the structure $\str A \setminus \str B$ contains
the relations in $\str A$ but not in $\str B$.
The treewidth of a set of atoms or a structure is the treewidth of
the respective Gaifman graph.

We sometimes write CQs $q$ as
$\anspred(\vec{x}) \leftarrow \calB$, where the body
$\calB = \{R_1(\vec{v}_1), \dots, R_m(\vec{v}_m)\}$
is a set of atoms and $\vec{x}$ are the {\em free variables}.
A Boolean \cq (BCQ) is a \cq with no free variables.
We define $\var(q) = \var(\calB)$. The existential variables
are implicitly given by $\var(\calB) \setminus \vec x$.
The result $q(\strD)$ of $q$ over a database $\strD$ is the set
of tuples $\{\mu(\vec x) \mid
\mu \colon \calB \ra \strD\}$.

\smallskip \noindent
{\bf Pattern trees (\pts).} 
  A {\em pattern tree} (short: \pt) $p$ over a relational schema $\sigma$ is
  a tuple $(T, \lambda, \calX)$ where $T$ is a rooted tree and $\lambda$ maps
  each node $t \in T$ to a set of relational atoms over $\sigma$.
  The set $\calX$ of variables denotes the {\em free variables} of the \pt. 
  We may write $((T,r), \lambda, \calX)$ to emphasize that $r$ is the root node of $T$.

  For a \pt $(T, \lambda, \calX)$ and a subtree $T'$ of $T$, we write $\lambda(T')$
to denote the set $\bigcup_{t \in V(T')}\lambda(t)$. %
We may write $\var(t)$ instead of $\var(\lambda(t))$, and $\var(T')$ instead of
$\var(\lambda(T'))$. We further define $\ifvar(t) = \var(t) \cap \calX$ and
$\ievar(t) = \var(t) \setminus \calX$ as the free and existential variables
in $T'$, respectively. These definitions extend naturally to subtrees $T'$
of $T$.
We call a \pt $(T, \lambda, \calX)$ \emph{projection free} if $\calX = \var(T)$.
We may write $(T, \lambda)$ to emphasize a \pt to be projection free.

We define the order $\prec$ among nodes $t \in T$ as 
$t_1 \prec t_2$ if $t_1$ is visited before $t_2$ in a
depth-first,left-to-right traversal of $T$. Also, for
$t \in T$, and a (not necessarily proper) subtree $T'$
of $T$, let $T'_{\prec t}$ be the subtree of $T$ that
contains all nodes $t' \in T'$ with $t' \prec t$.

\smallskip
\noindent
{\bf Semantics of \pts.} 
Evaluating a \pt $p$ with free variables $\calX$ over a database
$\strD$ returns a set $p(\strD)$ of mappings
$\mu \colon \calV \ra \idom(\strD)$ with $\calV \subseteq \calX$.
Intuitively, the idea of the evaluation is to evaluate the root
node first, resulting in a set of mappings. Then, in a top-down
left-to right traversal of the tree these mappings are extended
as far as possible by the solutions at the different nodes.
The semantics of $p(\strD)$ is, however, usually defined by
providing a characterization of the mappings generated by this
idea of a ``top-down evaluation''. We follow this approach and
use the characterization of solutions for weakly well-designed
pattern trees from \cite{KK2016} which also works as a definition
for the semantics of arbitrary \pts studied here.

\begin{definition}[(\cite{KK2016} pp-solution)]
 For a \pt $p = ((T,r), \lambda)$ and a database $\strD$, 
 a mapping $\mu \colon \calV \ra \idom(\strD)$ (with $\calV \subseteq
 \var(T)$) is a \emph{potential partial solution (pp-solution)}
 to $p$ over $\strD$ if there is a subtree $T'$ of $T$ containing $r$
 such that $\mu \colon \lambda(T') \ra \strD$.
\end{definition}
Observe that if $\mu$ is a pp-solution, then although this might be witnessed
by different subtrees $T'$, there exists a unique maximal such subtree $T'$.
We will denote it by $T_\mu$.
Also, for a mapping $\mu$ and some node $t \in T$, let $\mu_{\prec t}$ denote
the restriction $\mu|_{\var(T_{\prec t})}$.
\begin{definition}[($p(\strD)$)]\label{def:semPTs}
 Let $p = (T, \lambda, \calX)$ be a \pt and $p' = (T, \lambda)$ be
 the same but projection-free \pt. A mapping $\mu$ is in $p'(\strD)$
 if
 (1) $\mu$ is a pp-solution to $p'$ over $\strD$, and
 (2) there exists no child node $t'$ of $T_\mu$ and 
    homomorphism $\mu' \colon \lambda(t') \ra \strD$ extending $\mu_{\prec t'}$.
 For \pts with projection, we have 
 $p(\strD) = \{\mu|_{\calX} \mid \mu \in p'(\strD)\}$.
\end{definition}

\begin{example}\label{ex:PTsemantic}
 Consider the \pt $p_1$ in Figure~\ref{fig:relevantNodes}. Intuitively
 it asks for tickets $t$ and tries to assign seats to each ticket.
 In doing so, it first tries to find a seat in the ticket class (left
 child). If this is not possible, it tries to return any seat (right
 child). This reflects the intuitive semantics of pattern trees that
 nodes earlier in the order $\prec$ are evaluated first.
 
 Assume the database
 $\str D = \{{\it ticket}(1),$ ${\it class}(1,E),$ ${\it seatclass}(1,E),$ 
 ${\it seatclass}(2,F),$ ${\it empty}(1),$ ${\it empty}(2)\}$.
 The mapping $\mu = \{(x,1), (s,2), (c,F)\}$ is a pp-solution, as it maps
 the root and the second child node into $\strD$,
 and these two nodes contain exactly the variables in $\idom(\mu_1)$. 
 But $\mu$ is not ``maximal'' according to Definition~\ref{def:semPTs}.
 When testing for an extension to the first
 child, we may not test $\mu$, but $\mu|_{\{t\}}$, since $t$ is
 the only variable in $\idom(\mu)$ occurring in a node that precedes the
 first child in the order $\prec$.
 Thus $p_1(\strD) = \{\{(x,1), (s,1), (c,E)\}\}$.

\end{example}
\begin{figure}[t]
 \begin{center}
 \begin{tikzpicture}
  \node[] at (-4, 0) {$p_1\colon$};
  \node[draw, rounded corners=1mm] (r1) at (-1.6,0) { $ticket(t)$ };
  \node[draw, rounded corners=1mm] (t1) at (-3.1,-1) { \begin{minipage}{7.7em}
  	\centerline{${\it seatclass}(s,c)$}
  	\centerline{${\it empty}(s)$  ${\it class}(t,c)$}
  	\end{minipage}};
  \node[draw, rounded corners=1mm] (t2) at (-.3,-1) { \begin{minipage}{6em}
  	\centerline{${\it seatclass}(s,c)$}
  	\centerline{${\it empty}(s)$}
  	\end{minipage}};
  \draw (r1) -- (t1);
  \draw (r1) -- (t2);

   \node[] at (2, 0) {$p_2\colon$};
  \node[draw, rounded corners=1mm] (r1) at (4,0) { $a(x)$ };
  \node[draw, rounded corners=1mm] (t1) at (2.7,-1) { \begin{minipage}{6em}
  	\centerline{$c(y_i, y_j)$}
  	\centerline{$1 \leq i \neq j \leq n$}
  	\end{minipage}};
  \node[draw, rounded corners=1mm] (t2) at (5,-1) { \begin{minipage}{3em}
  	\centerline{$c(z_1, z_2)$}
  	\end{minipage}};
  \draw (r1) -- (t1);
  \draw (r1) -- (t2);

  \node[] at (6.6, 0) {$p_3\colon$};
  \node[draw, rounded corners=1mm] (r1) at (7.6,0) { $a(x)$ };
  \node[draw, rounded corners=1mm] (t1) at (6.9,-1) { \begin{minipage}{3em}
  	\centerline{$b(y)$}
  	\centerline{$bb_{1}(y_1)$}
  	\end{minipage}};
  \node[draw, rounded corners=1mm] (t2) at (8.3,-1) { \begin{minipage}{3em}
  	\centerline{$bt_{1}(y_1)$}
  	\centerline{$c(z)$}
  	\end{minipage}};
  \draw (r1) -- (t1);
  \draw (r1) -- (t2);

 \end{tikzpicture}
 \end{center}
 \caption{Example pattern trees referenced throughout the paper.}
 \label{fig:relevantNodes}
\end{figure}
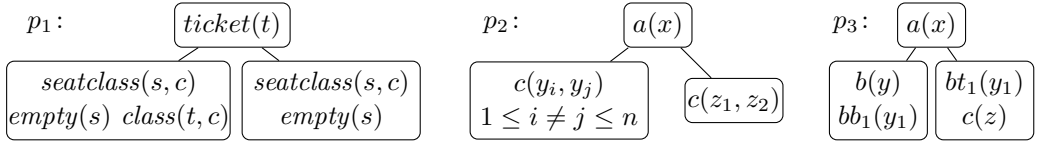

\smallskip \noindent
{\bf Parameterized complexity.} 
We only give a bare-bones introduction to parameterized complexity and refer the
reader to~\cite{FlumG06} for more details.
Let $\Sigma$ be a finite alphabet. A \emph{parameterization} of $\Sigma^*$
is a polynomial time computable mapping $\kappa \colon \Sigma^* \ra \calN$.
A \emph{parameterized problem} over $\Sigma$ is a pair ($L, \kappa$) where
$L \subseteq \Sigma^*$ and $\kappa$ is a parameterization of $\Sigma^*$.
We refer to $x \in \Sigma^*$ as the instances of a problem, and to the
numbers $\kappa(x)$ as the parameters.

A parameterized problem $E = (L, \kappa)$ belongs to the class $\FPT$ of
\emph{fixed-parameter tractable} problems if there is an algorithm $A$ 
deciding $L$, a polynomial $p$, and a computable function
$f \colon \calN \ra \calN$ such that the running time of $A$ is at most
$f(\kappa(x)) \cdot p(|x|)$. 

In this paper, for classes $\calP$ of \pts, we study the problem \peval{$\calP$}
defined below. We will also use the evaluation
problem \pcqeval{$\calQ$} for classes $\calQ$ of \cqs.

\begin{center}
\fbox{\hspace{-1mm}\begin{tabular}{ll}
 \multicolumn{2}{l}{\peval{$\calP$}} \\\hline 
 \\[-1ex]
  {\small INPUT}:&\hspace{-1em}A \wdpt $p \in \calP$, 
  	a database $\strD$, \\
  		&\hspace{-1em}and a mapping $\mu$.\\
  {\small PARAMETER}:&\hspace{-.8em}$|p|$\\
  {\small QUESTION}:&\hspace{-.8em}Does 
  	$\mu \in p(\strD)$ hold?\\
\end{tabular}\hspace{-1mm}}
\fbox{\hspace{-1mm}\begin{tabular}{ll}
 \multicolumn{2}{l}{\pcqeval{$\calQ$}} \\\hline 
 \\[-1ex]
  {\small INPUT}:&\hspace{-1em}A \cq $q \in \calQ$, a database $\strD$, \\
  	&\hspace{-1em}and a tuple $\vec v$.\\
  {\small PARAMETER}:&\hspace{-.8em}$|q|$\\
  {\small QUESTION}:&\hspace{-.8em}Does 
  	$\vec v \in q(\strD)$ hold?\\
\end{tabular}\hspace{-1mm}}
\end{center}

Let $E = (L,\kappa)$ and $E' = (L', \kappa')$ be parameterized
problems.
An \emph{$\FPT$-Turing reduction} from $E$ to $E'$ is an algorithm deciding $E$ with runtime 
at most $f(\kappa(x)) \cdot p(|x|)$ for a computable function $f$ and a polynomial $p$
where the algorithm can make calls to 
an oracle for $E'$ such that every oracle query $y$ satisfies $\kappa'(y)\le f(\kappa(x))$.
An $\FPT$-Turing reduction is called a \emph{many-one reduction}, if it makes only one oracle call
and accepts if and only if the answer to oracle call is positive.
Unless stated otherwise, all our $\FPT$-reductions will be many-one.
Of the rich parameterized hardness theory, we will only use the
class $\wone$ of parameterized problems. To keep this introduction short, we only
define a parameterized problem $(L, \kappa)$ to be $\wone$-hard if there is a problem $(L', \kappa')$
that reduces to $(L, \kappa)$. It is generally conjectured that $\FPT\ne \wone$ and thus in particular
$\wone$-hard problems are not in $\FPT$. We will take the hardness results on the problems $(L',\kappa')$ 
from the literature, often from the following result:
\begin{theorem}[(\cite{Grohe2007})] \label{theo:grohe}
 Let $\calQ$ be a decidable class of \emph{Boolean \cqs}.
 If there is a constant $c$ such that the treewidth of the core of each
 	\cq in $\calQ$ is bounded by $c$, then \pcqeval{$\calQ$} $\in$ $\FPT$. Otherwise,
 	\pcqeval{$\calQ$} is $\wone$-hard.
\end{theorem}

We will also study $\Hom(\calC)$ and $\pHom(\calC)$
for a class $\calC$ of structures. The input of both problems 
consists of a structure $\str A \in \calC$ and a structure $\str B$,
and the question is whether there exists a homomorphism 
$h \colon \str A \ra \str B$. For $\pHom(\calC)$, the parameter is
the size of $\str A$. 

\section{The Extension Problem}\label{sec:extproblem}

\begin{toappendix}

 \subsection{Additional Definitions}
 The proofs in this section require a few concepts that have already
 been introduced informally in the main part of this paper. In order
 to use these concepts in the fool proofs we first provide a formal
 definition for them.
 
 We start by recalling two well-known operators from the Relational Algebra:
 the projection $\pi$ and the selection $\sigma$. For a relation
 $R$ of arity $k$, we denote each of the $k$ positions of a tuple 
 $(a_1, \dots, a_k) \in R$ by their index $i$. Then the projection
 $\pi_{i_1,\dots, i_\ell}(a_1, \dots, a_k)$ returns the tuple
 $(a_{i_1}, \dots, a_{i_\ell})$, and for a relation $R$ the projection
 $\pi_{i_1,\dots, i_\ell}(R) = \{\pi_{i_1,\dots, i_\ell}(a_1, \dots, a_k)
 \mid (a_1, \dots, a_k) \in R\}$.
 For the selection $\sigma_{i_1=v_1, \dots, i_\ell=v_\ell}$, we
 say that a tuple $(a_1, \dots, a_k)$ satisfies the selection condition
 $({i_1=v_1, \dots, i_\ell=v_\ell)}$ if $a_{i_j} = v_j$ for all $1 \leq j
 \leq \ell$. For a relation $R$ we get
 $\sigma_{i_1=v_1, \dots, i_\ell=v_\ell}(R) = \{ (a_1, \dots, a_k) \in R
 \mid (a_1, \dots, a_k) \text{ satisfies } {(i_1=v_1, \dots, i_\ell=v_\ell)}\}$.

\medskip \noindent
\textbf{Projection of a Structure.}
Let $\str S$ be a $\sigma$-structure, and consider $V \subseteq \idom(\str S)$.
Then the \emph{projection of $\str S$ under $V$}, denoted by $\str S \setminus V$
returns the following $\sigma'$-structure (the set of relational symbols in
$\sigma'$ is defined implicitly by the newly introduced relations described below):
For every relation symbol $R \in \sigma$ and each $(a_1, \dots, a_k) \in R^{\str S}$,
let $o_1, \dots, o_{\ell}$ be those positions of $(a_1, \dots, a_k)$ that do not
contain values from $V$, i.e.\ $a_{o_j} \notin V$ for $q \leq j \leq \ell$, and
let $\{i_1, \dots, i_m\} = \{1, \dots, k\} \setminus \{o_1, \dots, o_\ell\}$
be those positions such that $a_{i_j} \in V$ for $1 \leq j \leq m$.

Then in structure $\str S \setminus V$ the relation
$$(R^{i_1=a_{i_1}, \dots, i_m=a_{i_m}})^{\str S \setminus V} 
\text{ contains the tuple } (a_{o_1}, \dots, a_{o_\ell}).$$

These are all the tuples in any relation in $\str S \setminus V$.

\medskip \noindent
\textbf{Projection of a Pair of Structures under a homomorphism.}
Next, let $(\str Q, \str D)$ be a pair of $\sigma$-structures and let
$h$ be a mapping on (a subset of) $\idom(\str Q)$.
We define the \emph{projection of $(\str Q, \str D)$ under $h$} as the
pair of $\sigma'$-structures $(\str {Q'}, \str D')$ as follows
(the vocabulary $\sigma'$ is again defined implicitly):

We start by defining $\str Q'$. For every relation symbol $R \in \sigma$
and every tuple $(a_1, \dots, a_k) \in R^{\str Q}$, we distinguish
two cases.
\begin{itemize}
	\item If $\{a_1, \dots, a_k\} \subseteq \idom(h)$ and 
		$(h(a_1), \dots, h(a_k)) \in R^{\str D}$, then ignore
		$(a_1, \dots, a_k)$ (i.e., no tuple derived from
		$(a_1, \dots, a_k)$ occurs in $\str Q'$).

	\item Otherwise, let 
		$\{i_1, \dots, i_\ell\} \subseteq \{a_1, \dots, a_k\}$
		be those positions of $(a_1, \dots, a_k)$ such
		that $a_{i_j} \in \idom(h)$ for $1 \leq j \leq \ell$,
		and $\{o_1, \dots, o_m\} = 	\{1, \dots, k\} \setminus
		\{i_1, \dots, i_\ell\}$ those positions such that
		$a_{o_j} \notin \idom(h)$.
		Then 
		$$ (R^{i_1=h(a_{i_1}), \dots, i_\ell=h(a_{i_\ell})})^{\str Q'} 
			\text{ contains the tuple }
		(a_{o_1}, \dots, a_{o_m}).$$
\end{itemize}
Observe that the second case includes the possibility that
$\ell = k$, i.e.\ that $a_i \in \idom(h)$ for all $1 \leq i \leq k$,
but that $(h(a_1), \dots, h(a_k)) \notin R^{\str D}$.
In this case, the result is a 0-ary relation symbol 
$R^{1=h(a_1), \dots, k=h(a_k)}$, and 
$(R^{1=h(a_1), \dots, k=h(a_k)})^{\str Q'}$ contains the empty
tuple $()$.

Next we define $\str D'$.
For every relation symbol $R^{i_1=b_1,\dots,i_\ell=b_\ell} \in \sigma'$
introduced by the definition of $\str Q'$ and every tuple
$(a_1, \dots, a_m) \in (R^{i_1=b_1,\dots,i_\ell=b_\ell})^{\str Q'}$,
let $k$ be the arity of the original relation symbol $R$ from $\sigma$.
Then the positions $1, \dots, m$ in $R^{i_1=b_1,\dots,i_\ell=b_\ell}$
correspond to positions $j_1, \dots, j_m$ in $R$. In fact, 
$\{j_1, \dots, j_m\} = \{1, \dots, k\} \setminus \{i_1, \dots, i_\ell\}$.
Then 
$$ (R^{i_1=b_1,\dots,i_\ell=b_\ell})^{\str D'} \text{ contains the tuples }
\pi_{j_1, \dots, j_m}(\sigma_{i_1=b_1,\dots,i_\ell=b_\ell}(R^{\str D}))$$

The following observation follows immediately from this definition.
\begin{observation} \label{obs:projectionhelps}
 Let two $\sigma$-structures $\str Q, \str D$ and a
 mapping $h \colon Q \ra \idom(\str D)$ be
 given where $Q \subseteq \idom(Q)$, and let 
 $\str Q', \str D'$ the projection of $(\str Q, \str D)$  under $h$.
 Then there exists some extension $h' \colon \str Q \ra \str D$ of $h$
 if and only if there exists a homomorphism $\hat h\colon \str Q' \ra \str D'$
 (i.e., if $(\str Q', \str D')$ is a positive instance of $\Hom$).
\end{observation}

We need to take care of one special case: Assume that, for a pair
$(\str Q, \str D)$ and a homomorphism $h$ such that $\str Q$ is already
a projection of some structure $\str L$ under a set $V \subseteq \idom(\str L)$,
we want to get the projection of $(\str Q, \str D)$ under $h$.
I.e.\ the vocabulary of $\str Q$ already contains relation symbols of the form
$R^{i_1 = b_1, \dots, i_\ell=b_\ell}$.
If for any of the $i_j$ ($1 \leq j \leq \ell$) it is the
case that $b_j \in \idom(h)$, then in the resulting
structure we replace $b_j$ in the name of the resulting
atom by $h(b_j)$.
In certain situations, this renaming of the relational symbols
will ensure the resulting structures to be over the same relational
schema, which is a prerequisite for finding homomorphisms. 

\subsection{Full Proofs of Section~\ref{sec:extproblem}}
\end{toappendix}

Definition~\ref{def:semPTs} shows that there are two potential sources
of complexity for the evaluation problem:
First, to determine whether some mapping is a pp-solution, and
second to check if a pp-solution is ``maximal''. 
As we will discuss in the next section, for projection free \pts,
the test for a pp-solution is easy. Hardness is thus
exclusively due to
testing maximality.

This problem is closely related to the homomorphism problem,
and can be easily reduced to it. However, done naively, this reduction loses the
information about the parts of the pattern tree already mapped into the database, which
might result in the reduction of an easy instance to a hard instance
of the homomorphism problem.

Thus, for a class $\calC$  of pairs of structure, in this section 
we study the following problem.
\begin{center}
\vspace{-1.3em}
\fbox{\hspace{-1mm}\begin{tabular}{ll}
 \multicolumn{2}{l}{\ext{$\calC$}} \\\hline 
 \\[-1ex]
  {\small INPUT}:&\hspace{-.7em}A pair $(\str A, \str B) \in \calP$ of
  	structures, a structure $\str C$,
  	and a homomorphism $h \colon \str A \ra \str C$. \\
  {\small QUESTION}:&\hspace{-.5em}Exists a homomorphism
  	$h' \colon \str B \ra \str C$ compatible with $h$? \\
\end{tabular}\hspace{-1mm}}
\end{center}
The problem \pExt{$\calC$} is the problem \ext{$\calC$} parameterized
by the size of $(\str A, \str B)$.

As mentioned earlier, the main difference between $\Hom$ and $\extp$
is that in addition to a structure, the input of $\extp$ gets another
structure and a homomorphism that is already guaranteed to map this
additional structure into the target. When looking for tractable 
classes, this additional input has to be taken into account.

To do so, we introduce the idea of the \emph{extension core}.
Towards its definition, for a set of elements $A$,
let $\str S_A$ be the
$\{R_a \mid a \in A\}$-structure 
(where each $R_a$ is a unique relation symbol) with 
$\idom(\str S_A) = A$ and $R_a^{\str S_A} = \{(a)\}$.

\begin{definition}[(Extension Core)]\label{def:extcore}
  Let $(\str A,\str B)$ be a pair of structures.
  The \emph{extension core $\extcore(\str A,\str B)$} 
  is the structure 
  $\extcore(\str A,\str B) = (\icore(\str A
  \cup \str B \cup \str S_{\idom(\str A)}) \setminus \str S_{\idom(\str A)})
  \setminus \idom(\str A)$.
\end{definition}
Said differently, the extension core is constructed 
by introducing a new relation for every domain element in
$\str A$ and then computing the core, the extension core accounts
on the one hand for the possibility that parts of $\str B$
might be folded into $\str A$ (and thus the extension of the
homomorphism to these parts is guaranteed), and on the other
hand for the fact that the mapping on $\idom(\str A)$ is fixed.
Removing $\idom(\str A)$ is then possible because the mapping
is already provided for these values.

The notion of the extension core allows us to formulate an exact
characterization of the tractable classes $\calC$ of the extension
problem \ext{$\calC$}. To this end, we define the treewidth of 
$\extcore(\calC)$ to be the maximum of the treewidth of $\extcore(\str A, \str B)$ 
for $(\str A, \str B)\in \calC$ if this maximum exists and $\infty$ otherwise.

\begin{theorem}\label{theo:extdichotomy}
 Assume that $\FPT \neq \wone$ and let $\calC$ be a
 decidable class of pairs of structures.
 Then the following statements are equivalent:
 \begin{enumerate}
  \item The treewidth of $\extcore(\calC)$ is bounded by a constant.
  \item The problem \ext{$\calC$} is in $\ptime$.
  \item The problem \pExt{$\calC$} is in $\FPT$.
 \end{enumerate}
\end{theorem}

The theorem is shown using a sequence of lemmas. However,
before the first lemma, we state an easy but important
observation that we use tacitly throughout this section.
\begin{observation}
	\begin{itemize}
		\item Extension cores are unique up to isomorphism. 
		\item For any two structures $\str A, \str B$, we have
		$\icore(\extcore(\str A, \str B)) = \extcore(\str A, \str B)$
	\end{itemize}
\end{observation}

The first result describes a crucial property of extension cores that
will be used several times throughout the remainder of this section.
\begin{lemma}\label{lem:extcoreEnough}
	An instance $(\str A, \str B), \str D, h$ of $\extp$ is
	a positive instance of $\extp$ if and only if there exists
	a homomorphism $h' \colon (\str A \cup \str S) \ra \str D$ that
	extends $h$, where 
	$\str S$ is the structure
	$\icore(\str A \cup \str B \cup \str S_{\idom(\str A)}) \setminus 
	\str S_{\idom(\str A)}$
	from the definition of extension cores.	
\end{lemma}
\begin{proof}
 Solving the instance $(\str A, \str B), \str D, h$ of $\extp$ is equivalent
 to solving the instance
 $((\str A \cup \str B \cup \str S_{\idom(\str A)}), 
 	(\str D \cup h(\str S_{\idom(\str A)}))$ of $\Hom$ 
 	(where $h(\str S_{\idom(\str A)})$ denotes the structure 
 	$\str S_{\idom(\str A)}$ where all elements $a \in \idom(h)$
 	are replaced by $h(a)$)
 which in turn is equivalent to deciding the existence of 
 $h' \colon (\str A \cup \str S) \ra \str D$ extending $h$.
\end{proof}
Next, we show the positive result, i.e.\ that the problem \ext{$\calC$}
can be solved efficiently if the treewidth of the extension cores in $\calC$
is bound.
\begin{lemmarep}\label{lem:exttwpositive}
 Let $\calC$ be a class of pairs of structures such that the treewidth of
 $\extcore(\str A, \str B)$ for all $(\str A, \str B) \in \calC$ is bounded
 by some constant $c$. Then \ext{$\calC$} is in $\ptime$.
\end{lemmarep}
\begin{inlineproof}[Proof (sketch)]
Given an instance $(\str A, \str B), \str D, h$ of \ext{$\calC$}, we know
from Lemma~\ref{lem:extcoreEnough} that we can equivalently solve the
problem whether there is an extension $h' \colon (\str A \cup \str S)
\ra \str D$ of $h$. This in turn can be shown to be equivalent to 
deciding an instance $(\str L, \str T)$ of $\Hom$, that we can derive
from $((\str A \cup \str S), \str D)$ (intuitively by replacing ``variables''
in $(\str A \cup \str S)$ according to $h$ by ``constants'' from
$\idom(\str D)$ and then reducing the resulting structure).
It can then be shown that when applying the same transformation
(replacing elements according to $h$ and reducing the structure)
to the pair $(\extcore(\str A, \str B), \str D)$, the resulting
structure is identical to $(\str L, \str T)$. Since 
$\extcore(\str A, \str B)$ has bounded treewidth and this transformation
does not increase the treewidth, it follows that the treewidth of
$\str L$ is bounded, and thus the instance $(\str L, \str T)$ can 
be solve in polynomial time \cite{DalmauKV2002}.
Finally, since all transformations can be computed in polynomial time as well,
we get the desired result.
\end{inlineproof}
\begin{proof}
 Let $(\str A, \str B) \in \calC$, $\str D$, and $h$ be an instance of
 \ext{$\calP$}. By Lemma~\ref{lem:extcoreEnough}, this problem is equivalent
 to asking for the existence of a homomorphism $h' \colon (\str A \cup \str S) \ra
 \str D$ that extends $h$, where $\str S = \icore(\str A \cup \str B \cup 
 \str S_A) \setminus \str S_A$ and $A = \idom(\str A)$.
 By Observation~\ref{obs:projectionhelps},
 this is equivalent to the instance $((\str A \cup \str S)', \str D')$ of
 $\Hom$, where  $((\str A \cup \str S)', \str D')$ is the projection of
 $((\str A \cup \str S), \str D)$ under $h$.

 For $\str E = \extcore(\str A, \str B)$ and $\str F = \str D$, we next show that
 $((\str A \cup \str S)', \str D') = (\str E', \str F')$ where 
 $(\str E', \str F')$ is the projection of $(\str E, \str F)$ under $h$.
 First of all, observe that $\str D' = \str F'$ does not necessarily holds,
 since it depends on the result of the projections of the left hand side.
 However, if the left hand side coincide, the equality $\str D' = \str F'$
 obviously holds. 

 We thus show that $(\str A \cup \str S)' = \str E'$. First of all, we have
 that 
 $\str A \cup \str S = \str A \cup ( \icore(\str A \cup \str B \cup \str S_{A})
 \setminus \str S_A) \subseteq (A \cup B)$. Moreover, since $h \colon \str A
 \ra \str D$, by the first case in the case distinction if the definition
 of the projection of a pair of structures under a homomorphism, 
 $(\str A \cup \str S)'$ does not contain any relation derived from $\str A$.
 It is thus safe to conclude that $(\str A \cup \str S)' = \str S'$, where
 $(\str S', \str D')$ is the projection of $(\str S, \str D)$ under $h$.

 Next, recall that
 $\str S = \icore(\str A \cup \str B \cup \str S_{A}) \setminus \str S_A$,
 and that
 $\extcore(\str A, \str B) = \str S \setminus A$. Thus the only difference
 between $\str E' = (\str S \setminus A)'$ and 
 $(\str A \cup \str S)' = \str S'$ is that in the first case, the projection
 of $\str S$ under $A$ is computed, before the projection under $h$.
 However, since $\idom(h) = \idom(\str A) = A$, it can be easily checked
 that this results in the same structures, and thus
 $\str E' = (\str A \cup \str S)'$. 

 Now the treewidth of $\str E$ is bounded by $c$, and therefore also the
 treewidth of $\str E'$ (taking subgraphs does not increase the treewidth).
 As a result, the existence of a homomorphism $\str E' \ra \str F'$ can
 be decided in polynomial time \cite{Grohe2007}, which proves the lemma.
\end{proof}

The next result shows that the above lemma is optimal by using the characterization
of tractable classes
(for both, $\ptime$ and $\FPT$) of $\pHom(\calC)$ provided in \cite{Grohe2007}.
\begin{toappendix}
\begin{lemma}\label{lem:hominSBij}
	Let $\str A$ and $\str B$ be structures and let $\str S$
	be the structure $\icore(\str A \cup \str B \cup \str S_A)$
	from the definition of the extension core. If $h$ is a homomorphism
	from $\str S$ to itself, then $h$ is bijective.
\end{lemma}
\begin{proof}
 Since $\str S$ is a core, any homomorphism from $\str S$ to
 itself is an isomorphism.
\end{proof}
\end{toappendix}

\begin{lemmarep}\label{lem:homToExt}
	Let $\calC$ be a decidable class of pairs of structures
	and let $\extcore(\calC)$ be the class of extension cores of the pairs
	in $\calC$. 
	Then $\pHom(\extcore(\calC))$ $\leq_\FPT$ \pExt{$\calC$}.
\end{lemmarep}
\begin{inlineproof}[Proof (idea)]
 Given an instance $(\str L, \str R)$ of $\pHom(\extcore(\calC))$, first
 compute a pair $(\str A, \str B) \in \calC$ such that
 $\extcore(\str A, \str B) = \str L$. 
 Next, we define a structure $\str D$ and a homomorphism $h \colon \str A \ra \str D$
 such that for $\str S = \icore(\str A \cup \str B \cup \str S_{\idom(A)}) \setminus 
 \str S_{\idom(A)}$ the following holds: There exists a homomorphism
 $h' \colon (\str A \cup \str S)
 \ra \str D$ extending $h$ if and only if there exists a homomorphism
 $\bar h \colon \str L \ra \str R$.
 By Lemma~\ref{lem:extcoreEnough} this implies that $(\str A, \str B), \str D, h$
 is a positive instance of \ext{$\calC$} if and only if $(\str L, \str R)$ is
 a positive instance of $\pHom(\extcore(\calC))$ and thus proves the case.
\end{inlineproof}
\begin{proof}
	Let $\calC$ be a decidable class of pairs of structures,
	and let $(\str L, \str T)$ be an instance of $\pHom(\extcore(\calC))$.
	We reduce this problem to an instance of \pExt{$\calP$}.

	In a first step, we compute a pair $(\str A, \str B) \in \calC$ such
	that $\extcore(\str A, \str B) = \str L$. We will define a structure
	$\str D$ and homomorphism $h \colon \str A \ra \str D$ such that
	there exists a homomorphism $h' \colon (\str A \cup \str B) \ra \str D$
	that is an extension of $h$ if and only if there exists a homomorphism
	from $\str L$ to $\str T$. However, we will define $\str D$ and 
	homomorphism $h$ in a slightly different setting.

	Let $\str S$ be the structure 
	$\icore(\str A \cup \str B \cup \str S_{A}) \setminus \str S_A$
	from the definition of extension cores where $A = \idom(\str A)$.
	By Lemma~\ref{lem:extcoreEnough}, the desired extension $h'$ of $h$
	exists if and only if there exists a homomorphism $h'' \colon
	(\str A \cup \str S) \ra \str D$ that extends $h$. Observe that the
	structure $\str D$ and homomorphism $h$ are still the same as above.
	We will thus work in the latter setting with $\str S$ instead of
	$\str B$ as this turns out to be easier.

	We define the structure $\str D$ over the same vocabulary as $\str S$ as follows:
	\begin{itemize}
		\item The domain $\idom(\str D) = \idom(\str T) \times \idom(\str S)$,
			i.e.\ the elements represent pairs of elements from
			$\str T$ and $\str S$, respectively. 

		\item For each relation symbol $R$ of arity $k$, and every
			tuple $(a_1, \dots, a_k) \in R^{\str S}$, the relation
			$R^{\str D}$ contains the following tuples:

			Let $\{i_1, \dots, i_{\ell}\} \subseteq \{1, \dots, k\}$
			be all those positions of $(a_1, \dots, a_k)$ such that
			$a_{i_j} \in \idom(\str A)$, and let $\{o_1, \dots, o_m\}
			= \{1, \dots, k\} \setminus \{i_1, \dots, i_{\ell}\}$ be
			all those positions such that $a_{o_j} \notin \idom(\str A)$,
			i.e.\ $a_{o_j} \in \idom(\str B) \setminus \idom(\str A)$.
			Let furthermore $R'$ be the relation symbol derived for
			$(a_1, \dots, a_k) \in R^{\str S}$ when computing the
			projection
			$\str S \setminus \idom(\str A) = \extcore(\str A, \str B)$.

			Now, for every pair $(\vec d_1, \vec d_2)$ of tuples 
			$\vec d_1 = (d_{o_1}, \dots, d_{o_m}) \in (R')^{\str T}$
			and 
			$\vec d_2 = (d_{i_1}, \dots, d_{i_\ell}) \in \idom(\str T)^{\ell}$,
			we add the tuple $((d_1, a_1), \dots, (d_k,a_k))$ to $R^{\str D}$.
			(Observe that by slight abuse of notation, in order to
			simplify the description we denote the positions in
			$\vec d_1$ and $\vec d_2$ according to the position
			in $R$ they originate from.) 
			Thus, intuitively, we replace all domain elements from 
			$\idom(\str A)$ with all possible combinations of elements from
			$\idom(\str T)$.
			
			These are all the tuples in $\str D$.
	\end{itemize}
	It is worth pointing out that in case $R'$ is not part of the vocabulary
	of $\str T$ or $(R')^{\str T}$ is empty, then by this definition $R^{\str D}$
	is the empty relation. The resulting instance will therefore be a simple 
	``no'' instance, because $R^{\str S}$ is non-empty. However, in this case
	we also have that $(R')^{\str L}$ is nonempty, and therefore also $(\str L, \str T)$
	is a trivial ``no' instance.
	
	Finally, we define the mapping $h \colon \idom(\str A) \ra \idom(\str D)$ as
	$h(a) = (d,a)$ for some arbitrary but fixed element $d \in \idom(\str T)$.

	It remains to prove that there indeed exists a homomorphism
	$g \colon \str L \ra \str T$ if and only if $h$ can be extended
	to a homomorphism
	$h' \colon \str S \ra \str D$.

	First assume that $g$ exists. Then define an extension $h'$ of $h$ 
	to $\idom(\str S)$ as $h'(a) = (g(a),a)$ for all
	$a \in \idom(\str S) \setminus \idom(\str A)$. The mapping $g$ is
	indeed defined on all these elements, since 
	$\idom(\str S) \setminus \idom(\str A) = \idom(\extcore(\str A, \str B))
	= \idom(\str L)$ because of $\extcore(\str A, \str B) = \str L)$.
	For $a \in \idom(\str A)$ we need not define $h'$ since $h$ is already
	defined on these elements, and $h'$ extends $h$.
	It now follows immediately from the construction of $\str D$ that $h'$
	is indeed the required homomorphism.

	For the other direction, assume that $h'$ exists. First, observe that
	$\str D$ projected onto the second component of its domain elements
	gives $\str S$. Thus $h'$ is a bijection in this second coordinate by
	Lemma~\ref{lem:hominSBij}. Let $\pi_2$ be the projection to the
	second coordinate. Then $\pi_2 \circ h$ is an automorphism of $\str S$,
	and thus there is a $n \in \mathbb{N}$ such that $(\pi_2 \circ h)^n = id$
	(where $id$ denotes the identity mapping).
	Consequently, w.l.o.g.\ we assume that $\pi_2 \circ h = id$.
	For every $a \in \idom(\str L) = \idom(\extcore(\str A, \str B))$ define
	$g(a)$ to be the value $d$ such that $h'(a) = (d,a)$. Then again by
	definition of $\str D$ it follows immediately that for all 
	relation symbols $R$ and tuples $\vec{a} \in R^{\str L}$ we have
	$g(\vec a) \in R^{\str T}$.

	Observing that all constructions can be done efficiently completes
	the proof.
\end{proof}

Theorem~\ref{theo:extdichotomy} now follows immediately. (1) $\Rightarrow$ (2)
follows from Lemma~\ref{lem:exttwpositive}. 
The implication (2) $\Rightarrow$ (3) follows immediately.
Finally, if the treewidth of $\extcore(\calP)$ is not bounded, then
$\pHom(\extcore(C))$ is not in $\FPT$ by \cite{Grohe2007}.
Thus, by Lemma~\ref{lem:homToExt}, the problem \pExt{$\calP$}
is not in $\FPT$, which shows (3) $\Rightarrow$ (1).

We will make use of $\extp$ and extension cores throughout the paper, but usually
for sets of atoms. In this case, we implicitly assume their common representation
as structures.

\begin{toappendix}
\subsection{Relationship to CQs}
Since it will turn out to be a useful tool in Section~\ref{sec:withproj},
and to substantiate our claim that $\extp$ is an interesting problem
on its own right, we observe the following relationship between
$\extp$ and \cqevalp which is reminiscent
of the relationship between the problems $\Hom$ and the evaluation problem
for Boolean \cqs.
For a \cq $q = \anspred(\vec x) \leftarrow \calB$, let
$\extcq(q) = (\{\anspred(\vec x)\}, \calB)$.
Furthermore, for
a class $\calQ$ of \cqs let $\extcq(\calQ) = \bigcup_{q\in \calQ} \extcq(q)$.
Then the following immediate corollary of Theorem~\ref{theo:extdichotomy}
provides a complete characterization via the treewidth of the extension
core of all tractable classes of \cqevalp\ over schemas with bounded arity
but an unbounded number of free variables.
\begin{corollary}\label{cor:nonBCQdichotomy}
 For every recursively enumerable class $\calQ$ of \cqs, the problems
 \cqeval{$\calQ$} and \ext{$\extcq(\calQ)$} are equivalent under many-one
 reductions.
\end{corollary}
\end{toappendix}

\section{Projection Free Pattern Trees}\label{sec:noproj}
We start with our investigation of \pts by looking at projection-free
\pts. As already mentioned at the beginning of the previous section,
given a \pt $p = ((T,r), \lambda)$, a database $\strD$, and a mapping
$\mu$, deciding whether $\mu$ is a pp-solution is feasible in
polynomial time. The algorithm could proceeds as follows:
First it identifies the set $N = \{ t \in T \mid \var(t) \subseteq \idom(\mu)
\text{ and } \mu(\tau) \in \strD \text{ for all } \tau \in \lambda(t)\}$.
Now if $r \notin N$, then clearly $\mu \notin p(\strD)$. 
Otherwise let $T'$ be the maximal subtree of $T$ that contains $r$ and
is built exclusively from nodes in $N$.
Then $\mu$ is a pp-solution if and only if $\var(T') = \idom(\mu)$.

Consequently, the $\coNP$-hardness of deciding $\mu \in p(\strD)$ originates
exclusively from testing whether $\mu$ can be extended.
Essentially, the reason for this test being hard is that it can be
the same %
as any homomorphism test.
However, testing the possibility of extending a mapping to a \emph{single}
node being hard does not necessarily make the complete problem of testing
the existence of any extension hard as well, as 
illustrated by the following example.

\begin{example}\label{ex:notAllChildren}
 Consider the \pt $p_2$ from Figure~\ref{fig:relevantNodes}.
 Let $\strD$ be some database
 containing at least $a(1)$ and assume the mapping $\mu = \{(x,1)\}$.
 Then $\mu$ clearly is a pp-solution. In this case testing whether
 $\mu \in p(\strD)$ boils down to deciding whether there exists either
 a mapping $\mu_1 \colon \lambda(t_1) \ra \strD$ or a mapping
 $\mu_2 \colon \lambda(t_2) \ra \strD$.
 Deciding the existence of $\mu_1$ is as hard deciding whether $\strD$
 contains a clique of size $n$. 
 However, observe the homomorphism $h \colon \lambda(t_2) \ra \lambda(t_1)$.
 Thus, whenever $\mu_1$ exists, $\mu_2$ exists as well.
 As a result, testing for the existence of a $\mu_1$ is not necessary.
 Instead, the easy test for $\mu_2$ is sufficient.
\end{example}

We formalize the observation of Example~\ref{ex:notAllChildren} by 
identifying for each subtree of a \pt exactly
those child nodes that potentially have to be tested for an extension.
Let $p = (T, \lambda)$ be a \pt and $T'$ be a subtree
of $T$. To start with, consider a set $C(T')$ of pairs of sets
of atoms 
$C(T') = \{ (\lambda(T_{\prec t_i}'),\lambda(t_i)) \mid t_i \in ch(T')\}$.
From this set of initial candidates, we eliminate redundant pairs as illustrated
in Example~\ref{ex:notAllChildren}.
That is, if there are two pairs $(\calT_i, \lambda(t_i)), (\calT_j, \lambda(t_j))
\in C(T')$ with $t_i \neq t_j$ such that there exists a homomorphism
$h \colon \calT_i \cup \lambda(t_i) \ra \calT_j \cup \lambda(t_j)$ with
$h|_{\var(\calT_i)}(x) = x$ for all $x \in \var(\calT_i)$,
remove $(\calT_j,\lambda(t_j))$ from $C(T')$.
Once there are no more such pairs, we denote the resulting set 
$C(T')$ as $\csts(T', T)$, and refer to its elements as 
\emph{\underline{c}ritical \underline{s}ub\underline{t}ree\underline{s}}.

Note that the construction of $\csts(T')$ is not deterministic because 
in each elimination step, there might be several choices
for elements to remove. Most of these choices lead to the same result,
since the composition of two homomorphisms is again a homomorphism.
However, when there are mutual homomorphisms between two pairs in $C(T')$,
then different choices may lead to different results.
Considering all possible elimination sequences, we thus get a set 
$\CST(T') = \{\csts_1(T'), \dots, \csts_\ell(T')\}$ of sets of critical subtrees.
As we will see, for our purposes all sets of critical subtrees in $\CST(T')$ are equivalent.

The first observation is a direct consequence of the definition of the sets
$\csts_i(T')$.
\begin{observation}\label{obs:cstForEachChild}
 For every child node $t \in ch(T')$ and every $\csts_i(T')$, either there
 is a pair $(\calT, \lambda(t)) \in \csts_i(T')$, or there is a node
 $t' \in ch(T')$ such that $(\calT', \lambda(t')) \in \csts_i(T')$ and
 there exists a homomorphism 
 $h \colon \calT \cup \lambda(t) \ra \calT' \cup \lambda(t')$ such that
 $h|_{\var(\calT)}$ is the identity mapping.
\end{observation}
The next property holds because the only reason for $\CST(T')$ containing
more than one element are mutual homomorphisms between two potential
candidates. This implies that both extension cores have the same treewidth. %
Recall that the treewidth of $\extcore(\csts_i(T'))$ is the maximal treewidth
of $\extcore(\lambda(T_{\prec t}'),\lambda(t))$ over $(\lambda(T_{\prec t}'),\lambda(t))\in \csts_i(T')$.
\begin{propositionrep}
 Let $(T, \lambda, \calX)$ be a \pt, $T'$ a subtree of $T$, and
 $\csts_i(T'), \csts_j(T') \in \CST(T')$.
 Then, for any $c \geq 1$, the treewidth of $\extcore(\csts_i(T'))$
 is less or equal than $c$ if and only if the treewidth of
 $\extcore(\csts_j(T'))$ is less or equal than $c$.
\end{propositionrep}
\begin{proof}
	As already mentioned in the main part of the paper, note that the
	composition of two homomorphisms is again a homomorphism. 
	Next, given the initial set $C(T')$, consider two sequences 
	$\delta_1, \delta_2$ of	deletions of pairs $(\calT_i, \lambda(t_i)$
	such that they result in different critical subtrees 
	$\csts_1(T')$ and $\csts_2(T')$. W.l.o.g.\ assume that there
	is a pair $(\calT_i, \lambda(t_i)) \in \csts_1(T') \setminus 
	\csts_2(T')$. Thus in $C(T')$, there was some 
	$(\calT_j, \lambda(t_j))$ that witnessed the deletion of
	$(\calT_i, \lambda(t_i))$ in the sequence $\delta_2$, i.e.\
	there exists a homomorphism $h \colon \calT_j \cup \lambda(t_j)
	\ra \calT_i \cup \lambda(t_i)$. However, since 
	$(\calT_i, \lambda(t_i))$ cannot be removed from $\csts_1(T')$,
	we have $(\calT_j, \lambda(t_j)) \notin \csts_1(T')$.
	However, the only pair that can witness the deletion of
	$(\calT_j, \lambda(t_j))$ without still witnessing (because
	of the composition of homomorphisms)
	a possible deletion of $(\calT_i, \lambda(t_i))$ is the pair
	$(\calT_j, \lambda(t_j))$ itself. We thus have both, a 
	homomorphism from $(\calT_j, \lambda(t_j))$ to 
	$(\calT_i, \lambda(t_i))$ and vice versa. 
	Since in addition these homomorphisms are the identity 
	on $\idom(\calT_i)$ and $\idom(\calT_j)$, respectively,
	they also give homomorphisms between the two extension cores.
	Thus, by \cite{Grohe2007} both extension cores are isomorphic,
	and therefore have the same treewidth.
\end{proof}
Thus in the following we ignore the ambiguities in the construction of 
the set of critical subtrees and let $\csts(T')$ be the result of an arbitrary run
of the construction algorithm. All our results
are invariant under this choice.
Finally, for a \pt $p = (T, \lambda)$, we define
$\csts(p) = \bigcup_{T' \text{ subtree of } T}\csts(T')$,
 and for a class $\calP$ of pattern trees $\csts(\calP) = \bigcup_{p \in \calP} \csts(p)$.

\begin{theorem}\label{theo:pf-evalextFPTequivalent}
 For every decidable class $\calP$ of projection free \pts,
 the problems \eval{$\calP$} and \ext{$\csts(\calP)$} are equivalent
 under $\FPT$-Turing reductions.
\end{theorem}
\begin{proof}
 The %
 reduction of \eval{$\calP$} to \ext{$\csts(\calP)$}
 is given by the following algorithm to decide $\mu \in p(\strD)$
for a mapping $\mu$, a database $\strD$, and a \pt $p = (T, \lambda)$:
First, decide whether $\mu$ is a pp-solution. If this is the case,
compute $T_\mu$ and $\csts(T_{\mu})$, otherwise return ``no''. 
Second, for every $(\calT, \lambda(t)) \in \csts(T_\mu)$, call an oracle for
$\extp$ on instance the $(\calT, \lambda(t)), \strD, \mu$. If
any of these calls returns ``yes'' $\mu$ is not maximal and thus we answer
``no'', otherwise we return ``yes''. The correctness of this algorithm follows
 from Definition~\ref{def:semPTs}, the discussion on critical subtrees and Observation~\ref{obs:cstForEachChild}.

 Next we show that \ext{$\csts(\calP)$} reduces to \eval{$\calP$}.
 Let $\calP$  be a class of \pts, and let
 $(\str A, \str B) \in \csts(\calP)$.
 Moreover, let $\str C$ be a structure over the same vocabulary $\sigma$ as
 $\str A \cup \str B$, and let $h \colon \str A \ra \str C$ be a homomorphism. Thus,
 $(\str A, \str B), \str C, h$ is an instance of \ext{$\csts(\calP)$}.
 We will show how to check whether $h$ can be extended to a homomorphism
 $\str A \cup \str B \ra \str C$ with the help of an oracle for \eval{$\calP$}.
 Towards this goal, we first find a projection free \pt $p = (T, \lambda)$
 with a subtree $T'$ of $T$ and a child node $t \in ch(T')$ such that
 $\lambda(T'_{\prec t}) \doteq \str A$, $\lambda(t) \doteq \str B$, and such that
 $(\lambda(T'_{\prec t}), \lambda(t)) \in \csts(T')$ (recall that we assume
 the implicit translation between sets of atoms and structures, indicated by
 $\doteq$).
 By definition, such a combination exists, and because $\calP$ is computable
 there is an algorithm to construct it.

 Next, we construct a new $\sigma$-structure $\str D$ with
 $\idom(\str D) = \idom(\str C) \times (\idom(\str A) \cup \idom(\str B))$:
 for every relation symbol $R \in \sigma$, the structure $\str D$
 contains the interpretation 
 $$R^{\str D} = \{((c_1,b_1), \dots, (c_\ell,b_\ell)) \mid 
 	(c_1, \dots, c_\ell) \in R^{\str C} \text{ and }
 	(b_1, \dots, b_\ell) \in R^{\str A \cup \str B} \}.$$

 Note that two homomorphisms 
 $h_{\str C} \colon (\str A \cup \str B) \ra \str C$ and
 $h_{\str A} \colon (\str A \cup \str B) \ra (\str A \cup \str B)$
 can be combined to a homomorphism 
 $h_{\str C \times \str A} \colon (\str A \cup \str B) \ra \str D$
 and that every homomorphism 
 $h_i \colon (\str A \cup \str B) \ra \str D$
 has such a representation as a product. 
 Clearly, constructing $\str D$ is in $\FPT$.

 We claim that $h$ can be extended to a homomorphism $h' \colon (\str A \cup 
 \str B) \ra \str C$ if and only if the mapping $\mu$ defined as
 $\mu = h \times id$ (where $id$ is the identity mapping on $\idom(\str A \cup
 \str B)$) is not an answer to $p$ on $\strD \doteq \str D$, i.e.\ if $\mu \notin p(\strD)$.

 To prove the claim, first assume that $\mu \in p(\strD)$.
 Observe that $\mu = \mu_{\prec t}$.
 Then by Definition~\ref{def:semPTs}, the mapping $\mu$ cannot be
 extended to a mapping $\mu'$ such that $\mu'(\tau) \in \strD$ for
 every atom $\tau \in  \lambda(t)$. But since 
 $id \colon (\str A \cup \str B) \ra (\str A \cup \str B)$ is a
 homomorphism, the the second component of $\mu$ could be extended
 to $\lambda(t) \doteq \str B$. 
 Hence the only reason why there is no such extension of $\mu$ is because
 of the first component. It thus follows that there cannot be a homomorphism
 $h' \colon (\str A \cup \str B) \ra \str C$ that extends $h$.
 This completes this direction of the proof.

 Next assume that $\mu \notin p(\strD)$. Clearly, $\mu(\tau) \in \strD$
 holds for all atoms $\tau \in \lambda(T') \doteq \str A$. Consequently, there
 must be a node $t' \in ch(T')$ such that there exists an extension
 $\mu_{\prec t'}'$ of $\mu_{\prec t'}$ with $\mu_{\prec t'}'(\tau) \in 
 \strD$ for all atoms $\tau \in \lambda(t')$. We claim that $t'$ must in
 fact be $t$.
 To see this, towards a contradiction, assume that $t' \neq t$, and that
 there exists such an extension $\mu_{\prec t'}'$.
 Then $\mu_{\prec t'}'$ decomposes into homomorphisms 
 $\mu_{\prec t'}' = h_{\str C} \times h_{\str A}$.
 Now $h_{\str A}$ is a homomorphism
 $h_{\str A} \colon \lambda(T'_{\prec t'}) \cup \lambda(t') \ra 
 \lambda(T') \cup \lambda(t) \doteq \str A \cup \str B$ that is the
 identity on $\var(T'_{\prec t'})$. This gives the desired
 contradiction, since the existence of $h_{\str A}$ would have lead to
 the elimination of $(\lambda(T'), \lambda(t))$ from $\csts(T')$.
 Thus $t = t'$ and there exists an extension $\mu'$ of $\mu$ with
 $\mu'(\tau) \in \strD$ for each atom $\tau \in \lambda(t)$.
 Again, $\mu'$ decomposes into homomorphisms $\mu' = h_{\str C} \times
 h_{\str A}$, and we have that $h_{\str C} \colon (\str A \cup \str B)
 \ra \str C$ is a homomorphism that extends $h$.
\end{proof}

Combining the above results, we thus get the following dichotomy for projection-free
\pts.
\begin{theorem}
Let $\calP$ be a decidable class of without projections and assume $\FPT\ne \wone$. 
Then \peval{$\calP$} is in $\FPT$ if and only if 
the treewidth of $\extcore(\lambda(T'),\lambda(t))$ for
all $(\lambda(T'), \lambda(t)) \in \csts(\calP)$ is bounded by a constant.
\end{theorem}

\section{Pattern Trees with Projection}\label{sec:withproj}

We now turn towards \emph{well-designed} pattern trees with projection.
As already mentioned in the Introduction, for technical reasons to be
discussed in Section~\ref{sec:conclusion}, in addition to the restriction
to well-designed pattern trees, we also restrict our study to \emph{simple}
pattern trees. However, most of the tractability results presented in
this section also hold for non-simple well-designed pattern trees. Even
more, all tractability results can be extended to arbitrary \pts. The
restriction to well-designed pattern trees is because we cannot characterize
the tractable classes of simple non well-designed pattern trees yet.

We start by defining \emph{simple} pattern trees, and then introduce
well-designed pattern trees.
\begin{definition}[(Simple PTs)]\label{def:simplePTs}
 A \pt $p = (T, \lambda, \calX)$ is a \emph{simple pattern tree}
 if at the one hand no atom $R(\vec v)$ occurs in 
 $\lambda(t)$ and $\lambda(t')$ for $t \neq t' \in T$, and on the
 other hand there are no two atoms $R_i(\vec v_i), R_j(\vec v_j) \in
 \lambda(T)$ with $R_i = R_j$.
\end{definition}

Well-designed pattern trees (\wdpts) restrict the distribution of variables
among the nodes of a pattern tree.
\begin{definition}[(Well-Designed Pattern Tree (\wdpt))]
 A \pt $(T, \lambda, \calX)$ is \emph{well-designed} if for every
 variable $y \in \var(T)$, the set of nodes of $T$ where $y$ appears
 is connected.
\end{definition}
An immediate consequence of this definition is that for every variable
$y \in \var(T)$, there exists a unique node $t \in T$ containing $y$
such that all nodes $t' \in T$ with $y \in \var(t')$ are descendants
of $t$. Because of this, the semantics of \wdpts can be described by
reformulating Definition~\ref{def:semPTs} in a much simpler way.
That is, given
a \wdpt $p = (T, \lambda, \calX)$ and a database $\strD$, a pp-solution
$\mu$ is in $p(\strD)$ if and only if there exists no child node $t'$
of $T_{\mu}$ and homomorphism $\mu' \colon \lambda(t') \ra \strD$
extending $\mu$.

The main result of this section will be a complete characterization of the
tractable classes of simple \wdpts. 
However, before we are able to do so, we first need to introduce yet another
property of (nodes of) \pts.
Pattern trees may contain nodes that are irrelevant for
query answering. Instead of resorting to normal forms to exclude pattern
trees containing such nodes like in \cite{LPPS2013}, here we may assume
\wdpts to contain such nodes.

In order to correctly characterize tractability, we must thus
be able to identify such nodes to exclude them from our tractability
considerations.
We formalize this concept of nodes potentially influencing the result
of a simple \wdpt after we introduced one more piece of notation necessary
for the definition. 
Let $(T,r)$ be a rooted tree and $t \in T$. Then $\branch(t)$ denotes
the set of nodes on the path from $r$ to the parent node of $t$. Moreover,
$\cbranch(t) = \branch(t) \cup \{t\}$.

\begin{definition}[(Relevant Nodes)]\label{def:relevantNode}
 Let $p = ((T,r), \lambda, \calX)$ be a simple \wdpt and $t \in T$.
 Then node $t \in T$ is \emph{relevant} if
 $\ifvar(T') \setminus \ifvar(\branch(t)) \neq \emptyset$ where
 $T'$ is the subtree of $T$ rooted in $t$.
 We use $\relvnode{T}$ to denote the set of relevant nodes in $T$.
\end{definition}

It follows immediately from \cite{LPPS2013} that this efficiently decidable
notion indeed captures the intended meaning.

\begin{proposition}
 Let $p = (T, \lambda, \calX)$ be a simple \wdpt. 
 A node $t \in T$ is relevant if and only if there exists a database $\strD$
 such that $p(\strD) \neq p'(\strD)$ where $p'$ is retrieved by removing the
 subtree of $p$ rooted in $t$.
\end{proposition}

With this notion settled, we can now aim towards our main result. The overall
idea of our algorithm for \peval{$\calP$} can be described as follows.
Given a \wdpt $p$, a database $\strD$, and a mapping $\mu$, we construct a
set of \cqs $q$ and associated databases $\strD'$ such that $\mu \in p(\strD)$
if and only if for at least one of these \cqs $q$ 
the tuple $\mu(\vec x)$ (where $\vec x \subseteq \idom(\mu)$ are the
free variables of $q$) is in $q(\strD')$. We will define three tractability
criteria that ensure this algorithm to be in $\FPT$. Intuitively, the third
tractability condition guarantees that deciding $\mu(\vec x) \in q(\strD')$
is in $\ptime$. The second condition guarantees the size of each 
relation in $\strD'$ to be polynomially bounded in the size of the input,
and the first condition guarantees that $\strD'$ can be computed efficiently.
Below, we will first introduce these three tractability criteria and show
that each of them is indeed necessary. Afterwards we show that they are also
sufficient by describing the $\FPT$ algorithm.

Note that some of the following notions and definitions are based on
ideas and similar notions in \cite{BPS2015} and \cite{KPS2016}.
However, many of them are refined carefully to provide a far more
fine-grained complexity analysis. 
Note further that the tractability results also work on \wdpts that
are not simple, and only the lower bounds hold for simple queries
only.

One important property of \pts that influences all three tractability
criteria are the (number and distribution of) variables that occur in
more than one node.
\begin{definition}[(\cite{BPS2015} Interface $\calI(t,t')$)]
 Let $(T, \lambda, \calX)$ be a \wdpt, $t, t' \in T$,
 and $T'$ a subtree of $T$.
 The \emph{interface $\calI(t,t')$} of $t$ and $t'$ is the
 set $\calI(t,t') = \var(t) \cap \var(t')$.
\end{definition}
While the size of the interface $\calI(t,t')$  (i.e., the number of
variables in each interface) for all pairs of nodes $t,t'$ can be
used for the definition of tractable classes (cf.\ \cite{BPS2015}),
this restriction is quite coarse. 
To provide a more fine grained tractability criteria, we first recall
the notion of an $S$-component from \cite{DurandM14b}:
Let $G = (V,E)$ be a graph, and $S \subseteq V$.
Then let $\calC$ be the set of connected components of $G[V \setminus S]$,
and for each $C \in \calC$, let $S_C \subseteq S$ be the set of nodes in
$S$ that have (in $G$) an edge to some node in $C$.
I.e., $S_C = \{ v \mid (v,v') \in E \text{ for some } v' \in C\}$.
The set $\calS$ of $S$-components of $G$ now is the set
$\{ G[C \cup S_C] \mid C \in \calC\}$.

For a set $S$ of variables, the notion of $S$-components extends to sets of
atoms in the obvious way via the Gaifman graph.
We will thus talk about $S$-components of sets of atoms. 

\begin{definition}[(\cite{KPS2016} Interface Components)]\label{def:iComponents}
 Let $p = ((T,r), \lambda, \calX)$ be a \wdpt, $t \in T$ a node of $T$ (but not the
 root node), and $\hat t$ the parent node of $t$.
 The set of \emph{interface components $\intcomp{t}$ of $t$} is a set of set of atoms,
 defined as the union of:
 \begin{enumerate}
 	\item The set
 	  $\{ \{\tau\} \mid \tau \in \lambda(t) \text{ and } \var(\tau) \subseteq 
 	  \calI(t,\hat t) \}$
 	  consisting of singleton sets for every atom $\tau \in \lambda(t)$ which contains
 	  only ``interface variables'', i.e.\ variables from $\calI(t,\hat t)$.

 	\item The set of all $\calI(t,\hat t)$-components of $\lambda(t)$.
 \end{enumerate}
 \end{definition}
Hence, interface components of ``type (1)'' are sets of single atoms,
while interface components of ``type (2)'' are sets of possibly several atoms.

To understand why interface components are essential for our results, recall
that solutions to \wdpts must be ``maximal'' (it must not be possible to extended
the mapping to some node).
Now a mapping cannot be extended to some node, if and only if it cannot be
extended to any one of its interface components. Thus instead of testing
the complete node at once (which might be intractable), we test the maximality
of a mapping component by component (which might be tractable). 

For each interface component, in general we are especially interested in the
existential interface variables occurring in it.
For a \wdpt $(T, \lambda, \calX)$ and a node $t \in T$ with parent node $\hat t$,
we therefore define the
\emph{inherited variables of an interface component $\calS \in \intcomp{t}$}
as the set $\calV_t(\calS) = (\calI(t,\hat t) \cap \var(\calS)) \setminus \calX$.

However, just for the first tractability condition the free variables are
actually of interest. Thus, for $((T,r), \lambda, \calX)$, $t$, $\hat t$,
and $\calS$ as before, let $\calV_t^+(\calS) = \calV_t(\calS) \cup (\ifvar(\hat t)
\cap \var(\calS))$.
Also, for a set $\calV$ of variables, recall the definition of the structure
$\str S_{\calV}$ from Section~\ref{sec:extproblem}.
\tractbox{a}{
	There is a constant $c$, such that for each $p = (T, \lambda, \calX) \in \calP$,
	the treewidth of $\extcore(\str S_{\calV_t^+(\calS)}, \calS)$ is bounded by $c$ 
	for all relevant nodes $t \in T$ (with $t \neq r$)
	and all $\calS \in \intcomp{t}$.
}

Intuitively, condition (a) guarantees that for each such interface component,
given some mapping $\mu$ on (a subset of) the free variables plus a mapping on
the inherited variables of this interface component, deciding
whether the mapping can be extended in such a way that all atoms of the interface
component are mapped into some database is in polynomial time.

Next, we show that tractability condition (a) is indeed necessary.
\begin{lemma}\label{lem:condAnec}
 Let $\calP$ be a decidable class of simple \wdpts such that 
 tractability condition (a) is not satisfied.
 Then \peval{$\calP$} is $\cowone$-hard.
\end{lemma}
\begin{proof}
 For a well-designed pattern tree $p$, let the \emph{relevant component set} 
 $\rcs(p)$ contain all the pairs $(\str S_{\calV_t^+(S)}, \calS)$
 as defined in tractability condition (a). Moreover, for a class
 $\calC$ of pattern trees, let $\rcs(\calC) = \bigcup_{p \in \calC}\rcs(p)$.
 We will --- by an $\FPT$-reduction --- reduce $\Hom(\extcore(\rcs(\calP)))$
 to the co-problem of \eval{$\calP$}. The result thus follows from
 \cite{Grohe2007}, since $\extcore(\rcs(\calP)) = \icore(\extcore(\rcs(\calP)))$
 does not have bounded treewidth by assumption.

 Thus, assume an instance $(\str E, \str F)$ of $\Hom(\extcore(\rcs(\calP)))$.
 First of all, find a \wdpt $p = ((T,r), \lambda, \calX) \in \calP$ with a relevant
 node $t \in T$ with $t \neq r$ and an interface component
 $\calS \in \intcomp{t}$ such that 
 $\str E = \extcore(\str S_{\calV_t^+(S)}, \calS)$.
 They exist by definition and since $\calP$  is decidable, finding
 them is in $\FPT$.

 Since $t$ is a relevant node, there exists at least one node $t' \in T$ with
 $\ifvar(t') \setminus \ifvar(\branch(t')) \neq \emptyset$ such that either
 $t = t'$ or $t'$ is a descendant of $t$. Among all nodes satisfying this
 property, pick $t'$ to be the node with the minimal distance to $t$.

 We define the set $N$ of nodes as follows. If $t = t'$, then $N = \emptyset$,
 otherwise set $N = \cbranch(t') \setminus \cbranch(t)$.

 We define a database $\str D$ over the set of relational symbols in $p$
 as follows. First, $\idom(\str D) = \idom(\str F) \cup \{d\}$,
 where $d$ is a fresh value not occurring in $\idom(\str F)$.
 The relations in $\str D$ contain the following tuples:
 \begin{itemize}
 	\item For each relation symbol $R$ mentioned outside of 
 		$\lambda(\cbranch(t'))$ set $R^{\str D} = \emptyset$.

 	\item For each relation symbol $R$ mentioned in $\lambda(\branch(t))$,
 		let $R^{\str D}$ contain the single tuple $(d, \dots, d)$.

 	\item For each relation symbol $R$ mentioned in
 		$\lambda(\cbranch(t') \setminus \branch(t)) \setminus \calS$, let
 		$R^{\str D}$ contain all possible tuples
 		$(a_1, \dots, a_k) \in \idom(\str F) \cup \{d\}$.
 		
 	\item For each relation symbol $R$ mentioned in $\calS$, observe
 		that there exists a relation symbol $R'$ in the vocabulary of
 		$\str E$ that was derived from $R$ during the computation of
 		the extension core. That is, the arity of $R'$ is less or equal
 		than the arity of $R$. Let $k$ be the arity of $R$, 
 		let $\{i_1, \dots, i_\ell\} \subseteq \{1, \dots, k\}$ be those
 		positions of $R$ containing values from $\calV_t^+(S)$,
 		and $\{o_1, \dots, o_m\} = \{1, \dots, k\} \setminus \{i_1, \dots, i_\ell\}$
 		those positions of $R$ that contain values from $\var(\calS)
 		\setminus \calV_t^+(S)$. 
 		Then, for every $(a_{i_1}, \dots, a_{i_{\ell}}) \in (R')^{\str F}$,
 		let $R^{\str D}$ contain all tuples 
 		$(a_1, \dots, a_k)$ with $(a_{o_1}, \dots, a_{o_m}) = (d, \dots, d)$,
 		i.e.\ we extend $(a_{i_1}, \dots, a_{i_{\ell}}) \in (R')^{\str T}$ by
 		assigning $d$ to the missing positions.

 \end{itemize}
 Finally, we define the last part of the instance of \eval{$\calP$},
 the mapping $\mu$ as $\mu(x) = d$ for all $x \in \ifvar(\branch(t))$.
 
 It is now the case that $\mu \in p(\str D)$ if and only if
 there is no homomorphism from $\str E$ to $\str F$. We prove
 this property in two steps. First, we show that $\mu \in p(\str D)$
 only depends on whether $\mu$ can be extended to $t$ or not.
 After this we show that such an extension of $\mu$ exists
 if and only if there is a homomorphism $h \colon \str E \ra \str F$.

 First, observe that the only possible extension $\mu'$ of $\mu$ such that
 $\mu'(\tau) \in \str D$ for every $\tau \in \lambda(\branch(t))$ is obviously
 $\mu'$ mapping every variable in $\var(\branch(t))$ to $d$. It follows 
 immediately, that for all nodes $t'' \neq t$ with their parent node in
 $\branch(t)$ the mapping $\mu'$ cannot be extended to $\lambda(t'')$, since
 for all relation symbols $R$ mentioned in $\lambda(t'')$ the relation
 $R^{\str D}$ is empty. 

 Next -- in order to conclude $\mu \in p(\str D)$ if and only if there exists
 no extension $\mu''$ of $\mu'$ with $\mu''(\tau) \in \str D$ for all $\tau
 \in \lambda(t)$ -- assume $\mu''$ exists. Then $\mu''$ can be obviously extended
 to $\mu'''$ with $\mu'''(\tau) \in \str D$ for all $\tau \in \cbranch(t')$ since
 for all atoms on $N \cup \{t'\}$, every possible atom over $\idom(\str D)$ is
 contained in $\str D$. Since the other direction -- $\mu'''$ exists, therefore
 there is an extension of $\mu'$ to $t$ -- is trivial, it remains to show that
 the extension $\mu''$ exists if and only if there is a homomorphism 
 $h \colon \str E \ra \str F$.

 To see that this is the case, observe that every such homomorphism $h$ in
 combination with $\mu'$ gives a mapping from $\calS$ into $\str D$, and vice
 versa, every mapping $\mu \colon \calS \ra \str D$ restricted to $\idom(\str F)$
 gives the desired homomorphism. For the remaining atoms in $\lambda(t) \setminus
 \calS$, observe that every possible mapping sends them into $\str D$, since
 $\str D$ again contains every possible atom for these relations.
\end{proof}

The second source of hardness are the existential variables shared between
a component and its predecessors, and the second condition restricts the
number of such shared variables.
\begin{definition}[(Interface Component Width)]
 Let $p = (T, \lambda, \calX)$ be a \wdpt, $t \in T$, and 
 $\calS \in \intcomp{t}$. 
 The \emph{width} of the interface component $\calS$ is $|\calV_t(\calS)|$. 
 For a node $t \in T$, the \emph{interface component width} of $t$ is
 the maximum width of any interface component $\calS$ of $t$.
 The  interface component width of $p$ is the maximum interface component
 width of all $t \in \relvnode{T}$.
\end{definition}

\tractbox{b}{ There is a constant $c$ such that for all $p \in \calP$
the interface component width of $p$ is at most $c$.}   

\begin{lemma}\label{lem:condBnec}
	Let $\calP$ be a decidable class of simple \wdpts such
	that tractability condition (b) is not satisfied.
	Then \peval{$\calP$} is $\cowone$-hard.
\end{lemma}
\begin{proof}
 We show the result by an $\FPT$-reduction of the problem of model checking
 FO sentences $\phi_k$ of the following form:
 \[ \phi_k = \forall x_1 \dots \forall x_k \exists y \bigwedge_{i \in k} E_i(x_i,y) \]
 By \cite{CD2012}, model checking for this class of sentences is
 $\wone$-hard.
 Therefore, let $\phi_k$ and a database $\str E$ with relations $E_i^{\str E}$
 over $\idom(\str E)$ be given.

 First, compute a \wdpt $p = (T, \lambda, \calX) \in \calP$ with an
 interface component width of at least $k$. W.l.o.g.\ we assume $p$
 to contain only binary atoms: Since we assume a bounded arity,
 binary atoms can be easily simulated with atoms of higher arity.
 Consider the relevant node $t \in T$ and an interface component $\calS \in \intcomp{t}$
 such that the interface component width of $\calS$ is at least $k$. 
 Since we assume relations to be of some bounded arity, $\calS$
 cannot be of type (1) (Definition~\ref{def:iComponents}).
 W.l.o.g., we thus assume that $\calS$ is of type (2).

 Since $t$ is relevant, there exists some $t'$ which is either $t$ itself or
 some descendant of $t$ such that $\ifvar(t') \setminus \ifvar(\branch(t')) \neq
 \emptyset$. Among all possible candidates, choose some $t'$ at a minimal
 distance to $t$.

 For the definition of the database $\strD$, recall that we assume each
 relation symbol to occur at most once in $p$. We define $\str D$ as
 follows.
 First, $\idom(\strD) = \idom(\str E)$. 
 Based on this, the database contains the following relations:
 \begin{itemize}
 	\item For each relation symbol $R$ mentioned outside of
 		$\lambda(\cbranch(t'))$ set $R^{\str D} = \emptyset$.

 	\item For each relation symbol $R$ mentioned in 
 		$\lambda(\cbranch(t')) \setminus \calS$,
 		let $R^{\str D}$ contain all possible tuples
 		$(a_1, \dots, a_k) \in \idom(\str E)$

 	\item For the relation symbols $R$ mentioned in $\calS$,
 		proceed as follows. Choose $k$ interface variables
 		$v_1, \dots, v_k \in \calI_t(\calS)$. Let $L = \var(\calS)
 		\setminus \var(\branch(t))$ be the ``local variables''  of $\calS$.
 		Observe that $L \neq \emptyset$, since otherwise $\calS$
 		could not be an interface component (it requires at least
 		one variable to connect the variables from $\calI_{t}(\calS)$).
 		By the same reasoning, for each of the variables $v_i$, 
 		there must exist at least one atom $R_i(v_i, z_i)$ or $R_i(z_i, v_i)$
 		for some $z_i \in L$. We will assume $R_i(v_i, z_i)$ in the following,
 		the other case is analogous). Now for each $v_i$, fix one such 
 		atom. Based on this, we define the following relations to be 
 		contained in $\str D$:

 		\begin{itemize}
 			\item For each of the selected atoms $R_i(v_i, z_i)$, let
 				$R_i^{\str D} = E_i^{\str E}$. I.e., we assume $R_i$
 				to take the place of $E_i$.

 			\item For every atom $R_i(z_1, z_2) \in \calS$ such that $z_1, z_2 \in L$,
 				define $R_i^{\str D}$ to contain all tuples
 				$\{(d,d) \mid d \in \idom(\str E)\}$.

 			\item For the remaining atoms $R_i(z_1, z_2) \in \calS$, define
 				$R_i^{\str D} = (\idom(\str E))^2$.

 		\end{itemize}
 \end{itemize}
 Then $\mu$ is an arbitrary mapping $\ifvar(\branch(t)) \ra \idom(\str E)$.

 It now follows by the same arguments as in the proof of Lemma~\ref{lem:condBnec}
 that we have $\mu \notin p(\str D)$ if and only if for every extension $\mu'$
 of $\mu$ to $\var(\branch(t))$, there exists an extension $\nu$ of $\mu'$
 such that $\nu(\tau) \in \str D$ for all $\tau \in \calS$.

 To close this proof, we thus only need to show that such an extension exists
 if and only if $\phi_k$ is satisfied. 

 First, assume that $\phi_k$ is satisfied. Then, for all $z \in L$, define
 $\nu(z)$ to be the value of $y$ in $\phi_k$. This clearly maps $\calS$
 into $\str D$. 

 Next, assume that $\phi_k$ is not satisfied. Thus there exists some assignment
 to $x_1, \dots, x_k$ such that no suitable $y$ value exists. Then for the mapping
 $\mu'$ assigning exactly those values to the selected interface variables 
 $v_1, \dots, v_k$, there exists no extension of $\mu'$ to $\calS$.
 This concludes the proof. 	
\end{proof}

It remains to define the tractability condition ensuring that the evaluation
problem for the resulting \cqs will be tractable. 
For this, we first need to introduce the notion of an \emph{component interface
atom}. For a well-designed pattern tree $(T, \lambda, \calX)$, a subtree $T'$ of $T$,
a node $t \in ch(T')$, and an interface component $\calS \in \intcomp{t}$, let the 
component interface atom $\cia{\calS}$ be the atom $R(\vec v)$ where $\vec v$
contains the variables in $\calV_t(\calS)$ and $R$ is a fresh relation symbol.
Observe that this implies $\cia{\calS} = R()$ in case $\calV_t(\calS) = \emptyset$.

The intuition for $R$ is that for each interface component, we want an atom that
covers exactly the variables in $\calV_t(\calS)$.

\tractbox{c}{ 
	There is a constant $c$ such that for every well-designed pattern tree
	$p = ((T,r), \lambda, \calX) \in \calP$ and every subtree $T'$ of $T$ containing
	$r$, the treewidth of 
	$\extcore(\str S_{\ifvar(T')}, \lambda(T') \cup 
		\bigcup_{i=1}^{n} \{\cia{\calS_i}\})$
	is bounded by $c$ for every combination
	$(\calS_1, \dots, \calS_n) \in \intcomp{t_1} \times \dots \times \intcomp{t_n}$
	where $\{t_1, \dots, t_n\} = ch(T') \cap \relvnode{T}$.	
}

\begin{lemma}\label{lem:condCnec}
	Let $\calP$ be a decidable class of simple \wdpts such
	that tractability condition (c) is not satisfied.
	Then \peval{$\calP$} is $\wone$-hard.
\end{lemma}
\begin{proof}
 First, assume that the interface component width of the instances is bounded.
 Otherwise, the result follows from Lemma~\ref{lem:condBnec}.
 In particular, we may thus assume that all instances of 
 $\extcore(\str S_{\ifvar(T')}, \lambda(T') \cup \bigcup_{i=1}^{n} \{\cia{\calS_i}\})$
 for all $p \in \calP$ are of bounded arity.

 Let $\textit{solcheck}(\calP)$ be the class of all structures
 $\extcore(\str S_{\ifvar(T')}, \lambda(T') \cup \bigcup_{i=1}^{n} \{\cia{\calS_i}\})$
 for $\calP$ as defined in tractability condition (c).
 We reduce \ext{$\mathit{solcheck}(\calP)$} to \peval{$\calP$} via an $\FPT$ reduction.
 The result then
 follows directly from Theorem~\ref{theo:extdichotomy}, since by assumption
 there does not exist a
 constant $c$ such that the treewidth of the extension core of each pair
 $(\str A, \str B) \in \mathit{solcheck}(\calP)$ is less or equal than $c$.

 Thus, let $(\str A, \str B)$, $\str D$, and $h$ be an instance of 
 \ext{$\textit{solcheck}(\calP)$}. We show how to construct a database $\strD'$
 and a mapping $\mu$ such that $\mu \in p(\strD')$ if and only if there exists
 an extension $h'$ of $h$ that is a homomorphism from $\str B$ to $\str D$.

 First of all, find a $p = (T, \lambda, \vec x) \in \calP$ and a subtree
 $T'$ of $T$ containing the root of $T$ such that $(\str A, \str B) =
 (\str S_{\ifvar(T')}, \lambda(T') \cup \bigcup_{i=1}^{n} \{\cia{\calS_i}\})$
 for some combination 
 $(\calS_1, \dots, \calS_n) \in \intcomp{t_1} \times \dots \times \intcomp{t_n}$
 where $\{t_1, \dots, t_n\} = ch(T') \cap \relvnode{T}$.
 For the definition of a database
 $\strD'$ and a mapping $\mu$ such that $\mu \in p(\strD')$ if and only if
 $h'$ exists, we need to define the following sets of nodes first.
 Let $ch(T') \cap \relvnode{T} = \{t_1, \dots, t_n\}$.
 For every $t_i \in ch(T')$, we define the set $N_i$ of nodes as follows:
 \begin{itemize}
 	\item If $\ifvar(t_i) \setminus \ifvar(\branch(t_i)) \neq \emptyset$
 		(i.e., $t_i$ contains some ``new'' free variable):
 		Then $N_i = \emptyset$.

 	\item Otherwise, let $s_i \in T$ be a descendant of $t_i$ such that
 		$\ifvar(s_i) \setminus \ifvar(\branch(t_i)) \neq \emptyset$ and
 		such that this property holds for no other node $s_i' \in T$
 		on the path from $t_i$ to $s_i$.
 		Then $N_i = \cbranch(s_i) \setminus \cbranch(t_i)$.
 \end{itemize}
 Finally, let $N = \bigcup_{i=1}^n N_i$.
 Now all notions are in place to describe the database $\strD'$. 

 For all atoms
 $R(\vec y) \in \lambda(T) \setminus (\lambda(T') \cup \lambda(ch(T')) \cup 
 \lambda(N))$,
 we set $R^{\strD'} = \emptyset$. I.e., for all atoms neither in $T'$ nor in any of
 the child nodes of $T'$ (or their extensions to some node with a ``new'' free
 variable), no matching values exist in the database.

 For all atoms $R(\vec y) \in \lambda(T')$, we do the following. 
 The relations for the atoms in $\lambda(T')$ are designed in such a way that all
 free variables $x \in \ifvar(T')$ in these atoms can only be mapped to $h(x)$.
 Since this way the values for all free variables are fixed,
 in the remainder, for atoms in $\lambda(T')$ we will only describe the values for
 those positions containing existentially quantified variables (recall
 that we only consider simple queries, thus these positions are uniquely defined).
 By slight abuse of notation (i.e., just ignoring the free variables), for every
 atom $R(\vec y) \in \lambda(T')$, we define $R^{\strD'} = R^{\strD}$.

 For all atoms $R(\vec y) \in \lambda(N)$, set
 $R^{\strD'} = \idom(\strD')^k$, where $k$ is the arity of $R$ (where $\idom(\str D')$
 is implicitly defined as the union over all elements mentioned in the definition of
 $\str D'$). 

 Finally, we define the relations for the atoms in $ch(T')$. Thus consider $t_i \in ch(T')$.
 Let $\vec v$ be the set of the inherited variables of the interface component
 $\calS_i \in \intcomp{t_i}$ selected for the construction of $\str B$. We use
 $R(\vec v)$ to denote the atom $\cia{\calS_i}$.

 For all atoms $R(\vec y) \in \lambda(t_i) \setminus \calS_i$,
 set $R^{\strD'} = \idom(\strD)^k$ where $k$ is the arity of $R$.
 For the atoms in $\calS_i$, we distinguish between
 $\calS_i$ being of type (1), or of type (2).

 If $\calS_i$ consists of a single atom of type (1), i.e.\ $\calS_i$ is of the form
 $R(\vec v) \in \lambda(t)$, define $R^{\strD'} = \idom(\strD)^k \setminus
 R^{\strD}$. 

 If $\calS_i$ is of type (2), we distinguish between two types of variables:
 Those that occur in $\calI_t(\calS_i)$, and those that are ``new'' in $\calS_i$,
 i.e.\ those that do not appear in some node $t' \in \branch(t_i)$.
 For these ``new'' variables, we will use as domain the set $\idom(\str D)^{|\vec v|}$,
 i.e.\ the set of all possible assignments of values from $\strD$ to the
 vector $\vec v$. Moreover, we assume some the ordering of the variables in $\vec v$
 to be in such a way that for a tuple $\vec a \in \idom(D)^{|\vec v|}$, the value
 at position $i$ corresponds to the variable $v_i \in \vec v$. For the remaining
 variables (i.e.\ the variables in $\vec v$), we will use values from $\idom(\str D)$.
 Then for each atom $R(\vec y) \in \calS_i$, the values in $R^{\str D'}$ are defined
 as follows.
 \begin{itemize}
 	\item If $\vec y \subseteq \vec v$, then $R^{\strD'} = \idom(\strD)^k$,
 		where $k$ is the arity of $R$.

 	\item If $\vec y$ contains ``new'' variables, i.e.\ 
 		$\vec y' = \vec y \cap (\var(\calS_i) \setminus \vec v) \neq \emptyset$,
 		then $R^{\strD'}$ contains all tuples satisfying the following
 		three properties.

 		\begin{enumerate}
 			\item all the ``new'' variables $\vec y'$ get assigned the
 				same value (say $\vec a \in \idom(\str D)^{|\vec v|}$),

 			\item for all variables $v_i \in \vec v \cap \vec y$, the value
 				of $v_i$ is consistent with $\vec a$ 
 				(say $\vec v \cap \vec y = \{v_{i_1}, \dots, v_{i_m}\}$
 				and let $\vec b$ be the values assigned to 
 				$\{v_{i_1}, \dots, v_{i_m}\}$), and

 			\item there exists no tuple $\vec r \in R^{\strD}$ such that
 				$\vec r$ projected onto $\{v_{i_1}, \dots, v_{i_m}\}$ is
 				$\vec b$ (i.e., $\vec b \notin \pi_{i_1, \dots, i_m}(R^\strD)$).
 		\end{enumerate}
 \end{itemize}

 Note that all of these relations can be constructed in polynomial time
 because the arity of all relations is assumed to be bound.

 We claim that, with the definition above, for an assignment $\mu'$ to
 $\vec v$, we have $(\mu',\strD) \models R(\vec v)$ (i.e., $R(\mu'(\vec v)) \in
 \strD$) if and only if all extensions $\mu''$ of $\mu'$ to $\var(\lambda(t_i))$
 do not map all atoms in $\calS_i$ into $\strD'$,
 and thus also not $\lambda(t_i)$ (since $\calS_i \subseteq \lambda(t_i)$).
 To see this, first assume that $(\mu', \strD) \models R(\vec v)$. Let
 $\mu''$ be an extension of $\mu'$ to $\var(\calS_i)$ that satisfies conditions
 1. and 2. Then all variables in $\var(\calS_i) \setminus \vec v$ take the same
 value under $\mu''$, and this value is exactly the tuple $\mu'(\vec v)$.
 But then $\mu''$ does not satisfy condition 3., because $(\mu',\strD) \models
 R(\vec v)$ (This implies $\mu'(\vec v) \in R^{\strD}$, and thus
 $\mu'(\vec v)$ projected onto the variables in $\vec v \cap \vec y$
 gives exactly $\mu''_{|\vec v \cap \vec y}(\vec v \cap \vec y)$, contradicting
 3.). So $\mu''$ does not map all atoms in $\calS_i$ into $\strD'$.
 For the other direction,
 assume that no extension $\mu''$ of $\mu'$ maps all atoms in $\calS_i$ into
 $\str D'$.
 Then this is in particular true for those assignments satisfying 1. and
 2. Consequently, $\mu'_{|\vec v} = \mu'$ does not map all atoms in $\calS_i$
 on $\strD$ due to 3. But this means that $(\mu', \strD) \models R(\vec v)$,
 proving the claim.

 We now claim that the assignment $\mu$ setting all free variables $x$ in $T'$
 to $h(x)$ is an answer to $p$ on $\strD'$ if and only if the required extension
 $h'$ of $h$ exists.
 First, observe that $\mu \in p(\str D')$ if and only if 
 \begin{itemize}
 	\item there is an extension $\mu'$ of $\mu$ to $\var(T')$ that maps
 		all atoms in $\lambda(T')$ into $\strD'$
 		(of course, in general every subtree $T''$ with $\ifvar(T') = \idom(\mu)$
 		is a potential candidate, but given the construction of $\str D'$, the
 		subtree $T'$ is the only possible candidate)
 		and 

 	\item for all $t_i \in ch(T')$, we have that there does not
 		exists an extension of $\mu'_{t_i}$ onto $\lambda(t_i) \cup \lambda(N_i)$.
 		(In fact, extending the mapping to any descendant of $t_i$ that
 		contains some additional free variable would work. However, the
 		only nodes with non-empty relations in $\str D'$ are those mentioned
 		in $N$.)

 \end{itemize}

 By the construction of $\strD'$ for atoms in $\lambda(N)$, for every
 $t_i \in ch(T')$ it follows immediately that there exists an extension
 of $\mu'_{t_i}$ onto $\lambda(t_i) \cup \lambda(N_i)$ if and only if
 there exists an extension to $\lambda(t_i)$: Since for the atoms in 
 $\lambda(N)$ the database $\strD'$ contains all possible tuples,
 every extension $\mu''$ of $\mu'_{t_i}$ onto $\lambda(t_i)$ can be
 further extended to all atoms in $\lambda(N_i)$.

 Note that the existence of an extension of $\mu'_{t_i}$ onto $\lambda(t_i)$
 is, as we have seen before, equivalent to $\mu'$ satisfying $R(\vec v)$,
 the atom introduced for the node $t_i$ in $q$.
 So $\mu \in p(\strD')$ if and only if there is an extension $h'$ of $h$
 that is a homomorphism from $\str B$ into $\str D$. This completes the proof.
\end{proof}

Having defined the three tractability conditions, we next show how they
may be used to design an $\FPT$-algorithm for \eval{$\calP$}.
\begin{algorithm}
  \caption{EvalFPT($p$, $\strD$, $\mu$)}
  \begin{algorithmic}[1]\small
  	\State $T = T[\relvnode{T}]$ \Comment Remove all nodes from $T$ that are not relevant
    \ForAll {subtrees $T'$ of $T$ with $\ifvar(T') = \idom(\mu)$}
      \State Let $\{t_1, \dots, t_n\} = ch(T') \cap \relvnode{T}\}$ 
      \ForAll {$(\calS_1, \dots, \calS_n) \in 
      	\intcomp{t_1} \times \dots \times \intcomp{t_n}$}
      	\State $q$ $=$ ``$\anspred(\vec x) \leftarrow \lambda(T') \cup 
      		\{\cia{\calS_1}, \dots, \cia{\calS_n}\}$''
        	\Comment Let $\vec x$ contain all $x \in \ifvar(T')$
        \State $\strD'$ $=$ $\strD \cup \bigcup_{i=1}^{n} 
        	\{R_i(\nu(\vec v_i)) \mid \nu \in \stopmaps(\calS_i, \strD)\}$
        	\Comment Assume $\cia{T'}{\calS_i} = R_i(\vec v_i)$
        \If {$ \mu(\vec x) \in q(\strD')$} \Call{exit}{YES}
        \EndIf
      \EndFor    
    \EndFor
    \State \Call{exit}{NO}
\end{algorithmic}
\label{alg:fptalgo}
\end{algorithm}
The algorithm is outlined in Algorithm~\ref{alg:fptalgo}.

First of all, we discuss the only notion in Algorithm~\ref{alg:fptalgo} not yet
introduced in this section: $\stopmaps(\calS,\strD)$ for an interface component
$\calS$ and a database $\strD$. Recall that we said earlier that the intention
of the interface components is to ensure a mapping to be maximal not by testing
for extensions to the complete node, but to do these tests component wise.

The idea how to realize this is to store in $\strD'$ for each interface component $\calS$
those variable assignments $\nu$ to its inherited variables such that there
exists no extension $\nu' \colon \calS \ra \strD$ of $\nu \cup \mu$. 
These are exactly the values stored in $\stopmaps(\calS_i, \strD)$.

Formally, for a \wdpt $(T, \lambda, \calX)$, a subtree $T'$ of $T$, a child node
$t \in ch(T')$, an interface component $\calS \in \intcomp{t}$, a database
$\strD$, and a mapping $\mu \colon \calX' \ra \idom(\strD)$ (for $\calX' \subseteq
\calX)$), 
consider the mappings $\extmaps(\calS, \strD) =
\{ \eta|_{\calV_t(\calS)} \mid \eta \colon \var(\calS) \ra \idom(\strD), 
 \eta \text{ extends } \mu|_{\var(\calS)} \text{ and }
 \eta(\tau) \in \strD \text{ for all } \tau \in \calS \}$. 
I.e., $\extmaps$ contains exactly those mappings on $\calV_t(\calS)$
that \emph{can} be extended in a way that maps $\calS$ into $\strD$.
We thus get 
$\stopmaps(\calS, \strD) = \{ \nu \colon \calV_t(\calS) \ra \idom(\strD)
\mid \nu \notin \extmaps(\calS, \strD)\}$.

With this in place, we describe the idea of Algorithm~\ref{alg:fptalgo}.
Recall that given $\mu$, we have to find a mapping $\mu'$ extending $\mu$
that on the one hand (1) is a pp-solution, and on the other hand (2) is maximal.
Unlike the case without projection, where $T_\mu$ is easy to find, because
of the presence of existential variables, there might be up to an exponential
number of candidates for $T_\mu$: all subtrees $T'$ of $T$ with $\ifvar(T') 
= \idom(\mu)$. After removing all irrelevant nodes (they might make evaluation
unnecessarily hard), we thus check all of these candidates. 

If the required mapping $\mu'$ exists then, as discussed earlier, for
each child node of $T'$ there exists at least one interface component
to which $\mu'$ cannot be extended. Not knowing which interface components
these are, the algorithm iterates over all possible combinations (line 4).
In lines 5--7, the algorithm now checks whether there exists an extension
of $\mu$ that maps all of $\lambda(T')$ into $\strD$ (ensured by adding
$\lambda(T')$ to $q$), but none of the interface components
$\calS_1, \dots, \calS_n$. The latter property is equivalent to asking that
for each $\calS_i$, $\mu'$ must assign a value to its inherited interface
variables that cannot be extended. This is guaranteed by adding the
atoms $\cia{\calS_i}$ to $q$ and providing in $\strD'$ exactly the
values from $\stopmaps(\calS_i, \strD)$. 

In order to see that this indeed gives an $\FPT$ algorithm in case
tractability conditions (a), (b), and (c) are satisfied, note that
condition (b) ensures that the arity of each of the new relations for
the atoms $\cia{\calS}$ is at most $c$. Thus the size of these
relations (and thus the number of possible mappings in $\stopmaps(\calS,\strD)$)
is at most $|\idom(\strD)|^c$. Next, condition (a) ensures that for each
mapping $\nu \colon \calV_t(\calS) \ra \idom(\strD)$ deciding membership
in $\stopmaps(\calS,\strD)$ is in $\ptime$. Finally, condition (c) ensures
that the test in line~7 is feasible in polynomial time.

We note that the algorithm is an extension and refinement of the $\FPT$ algorithm
presented in \cite{KPS2016}. An inspection of \cite{KPS2016} reveals that
the conditions provided there imply our tractability conditions (a), (b),
and (c), but not vice-versa. 
In fact, our conditions explicitly describe the crucial properties of
their restrictions that allow the problem to be in $\FPT$.

From Algorithm~\ref{alg:fptalgo} we thus derive the following result.
\begin{theorem}\label{theo:fptwp}
	Let $\calP$ be a decidable class of simple \wdpts.
	If the tractability conditions (a), (b), and (c) hold for $\calP$,
	then \peval{$\calP$} can be solved in $\FPT$.
\end{theorem}
The correctness of the algorithm follows immediately from the above discussion.
For the runtime, observe that in addition to what we already discussed, the number
of loop-iterations in lines 2 and 6 are bounded in the size of $p$, which is the
parameter for the problem.

Combining Theorem~\ref{theo:fptwp} with Lemmas~\ref{lem:condAnec}, 
\ref{lem:condBnec}, and \ref{lem:condCnec}, we get the following
characterization.
\begin{theorem}\label{theo:dichoWithProjection}
  Assume that $\FPT \neq \wone$, and let $\calP$ be a decidable class of simple \wdpts.
  Then the following statements are equivalent.
  \begin{enumerate}
    \item The tractability conditions (a), (b), and (c) hold
      for $\calP$.
    \item \eval{$\calP$} is in $\FPT$.
  \end{enumerate}
\end{theorem}

We mentioned at the beginning of this section that most tractability results
also hold for arbitrary \wdpts instead of just simple ones. Observe that we make
use of simple \wdpts only in Lemmas~\ref{lem:condAnec}, \ref{lem:condBnec}, and
\ref{lem:condCnec}, as well as in the definition of relevant nodes.
In fact, whenever a arbitrary well-designed pattern tree satisfies the tractability
conditions for all nodes (instead of just the relevant ones), Algorithm~\ref{alg:fptalgo}
also provides an $\FPT$ algorithm for the evaluation problem.

\section{Relationship with SPARQL and Conclusion}
\label{sec:conclusion}
The results of Sections~\ref{sec:noproj} and \ref{sec:withproj} give
a fine understanding of the tractable classes of \pts without projection
and \wdpts in the presence of projection. 
In particular they show the different sources of hardness. 
As laid out in the introduction, there is a strong relationship between
(weakly) well-designed SPARQL queries and classes of \pts, namely the
weakly well-designed (\wwdpt) and well-designed (\wdpts) pattern trees,
respectively: For every (weakly) well-designed SPARQL query, an equivalent
(weakly) well-designed pattern tree can be computed in polynomial time,
and vice versa, in a completely syntactic way.

Since our results for projection-free \pts apply to all classes of \pts,
they therefore immediately apply to (weakly) well-designed \{{\sf AND},
{\sf OPTIONAL}\}-SPARQL queries as well. 
Note that the correspondence is unfortunately less tight for the case with
projections.
Not only because we study only well-designed pattern trees instead
of arbitrary ones, but recall that our characterization only applies for
classes of \emph{simple} well-designed pattern trees. 
However, RDF triples and SPARQL triple patterns, in the relational
model, are usually represented with a single (ternary) relation. 
Thus, there is no direct translation to and from simple (well-designed)
pattern trees. As a consequence, in the presence of projection, our
characterization of tractable classes of simple well-designed pattern trees
does not imply an immediate
characterization of the tractable classes of well-designed 
\{{\sf AND}, {\sf OPTIONAL}\}-SPARQL queries.

Nevertheless, our results also give interesting insights to SPARQL with projections.
First of all, Algorithm~\ref{alg:fptalgo} can directly be applied without any 
changes also for queries in which relation symbols can appear several times and 
thus in particular for well-designed pattern trees that result from the translation
of well-designed SPARQL queries.
Moreover, our result determines completely the tractable classes that can be 
characterized by analyzing only the underlying graph structure of the 
queries, i.e., the Gaifman graph. Indeed, since simple queries can simulate 
all other queries sharing the same Gaifman graph by duplicating relations, 
Gaifman graph based techniques have exactly the same limits as simple queries.
Thus, our work gives significant information on limits of tractability 
for SPARQL queries in the same way as e.g.~\cite{GroheSS01,Chen14} did in 
similar contexts.

Let us mention the major stumbling block towards a characterization of 
non-simple well-designed pattern trees with projections: In the proof of
Lemma~\ref{lem:condBnec}, 
we have used a reduction from quantified conjunctive queries.
Unfortunately, the tractable classes for the non-simple fragment and the correct notion of cores
for that problem are not well understood which limits our result to simple queries since we are using the respective 
results from~\cite{CD2012}. Note that we might have been able to give a 
more fine-grained result in sorted logics by using~\cite{ChenM13}, but since 
this would, in our opinion, not have been very natural in our setting, we did not pursue 
this direction. Thus getting an even better understanding of non-simple pattern trees
would either need progress on quantified conjunctive queries or a reduction 
from another problem that is better understood.

For future work, there are further SPARQL features
that we did not include in the framework studied in this paper. The most prominent
among them are of course {\sf FILTER} expressions. Let us remark that pattern trees with
{\sf FILTER} expressions are easily seen to subsume conjunctive queries with inequalities---just
consider patter trees consisting of only one node---and thus in particular also graph embedding problems.
The tractable fragments of the latter are a notorious problem that has resisted 
the efforts of the parameterized complexity community for a long time now, even 
though there has recently been progress in the area, see e.g.~\cite{Lin15,ChenGL17}. Thus
showing a complete characterization of the tractable classes for SPARQL queries 
with {\sf FILTER} is probably very hard. Still, it would be interesting to
give algorithms extending our results to that fragment and maybe giving lower bounds
based on the conjectured dichotomy for embedding problems.

\newpage

\bibliography{references}

\begin{thebibliography}{10}

\bibitem{AHV}
Serge Abiteboul, Richard Hull, and Victor Vianu.
\newblock {\em Foundations of Databases}.
\newblock Addison-Wesley, Reading, Massachusetts, 1995.

\bibitem{AFPSS2015}
Shqiponja Ahmetaj, Wolfgang Fischl, Reinhard Pichler, Mantas Simkus, and
  Sebastian Skritek.
\newblock Towards reconciling {SPARQL} and certain answers.
\newblock In {\em Proc.\ {WWW} 2015}, pages 23--33. {ACM}, 2015.

\bibitem{AG2008}
Renzo Angles and Claudio Gutierrez.
\newblock The expressive power of {SPARQL}.
\newblock In {\em Proc.\ {ISWC} 2008}, volume 5318 of {\em Lecture Notes in
  Computer Science}, pages 114--129. Springer, 2008.

\bibitem{AG2016}
Renzo Angles and Claudio Gutierrez.
\newblock The multiset semantics of {SPARQL} patterns.
\newblock In {\em Proc.\ {ISWC} 2016}, volume 9981 of {\em Lecture Notes in
  Computer Science}, pages 20--36. Springer, 2016.

\bibitem{AG2016b}
Renzo Angles and Claudio Gutierrez.
\newblock Negation in {SPARQL}.
\newblock In {\em Proc.\ {AMW} 2016}, volume 1644 of {\em {CEUR} Workshop
  Proceedings}. CEUR-WS.org, 2016.

\bibitem{AACP2013}
Carlos~Buil Aranda, Marcelo Arenas, {\'{O}}scar Corcho, and Axel Polleres.
\newblock Federating queries in {SPARQL} 1.1: Syntax, semantics and evaluation.
\newblock {\em J. Web Sem.}, 18(1):1--17, 2013.

\bibitem{ADK2016}
Marcelo Arenas, Gonzalo~I. Diaz, and Egor~V. Kostylev.
\newblock Reverse engineering {SPARQL} queries.
\newblock In {\em Proc.\ {WWW} 2016}, pages 239--249. {ACM}, 2016.

\bibitem{AP2011}
Marcelo Arenas and Jorge P{\'{e}}rez.
\newblock Querying semantic web data with {SPARQL}.
\newblock In {\em Proc.\ {PODS} 2011}, pages 305--316. {ACM}, 2011.

\bibitem{AU2016}
Marcelo Arenas and Mart{\'{\i}}n Ugarte.
\newblock Designing a query language for {RDF:} marrying open and closed
  worlds.
\newblock In {\em Proc.\ {PODS} 2016}, pages 225--236. {ACM}, 2016.

\bibitem{BPS2015}
Pablo Barcel{\'{o}}, Reinhard Pichler, and Sebastian Skritek.
\newblock Efficient evaluation and approximation of well-designed pattern
  trees.
\newblock In {\em Proc.\ {PODS} 2015}, pages 131--144. {ACM}, 2015.

\bibitem{CEGL2012b}
Melisachew~Wudage Chekol, J{\'{e}}r{\^{o}}me Euzenat, Pierre Genev{\`{e}}s, and
  Nabil Laya{\"{\i}}da.
\newblock {SPARQL} query containment under {RDFS} entailment regime.
\newblock In {\em Proc.\ {IJCAR} 2012}, volume 7364 of {\em Lecture Notes in
  Computer Science}, pages 134--148. Springer, 2012.

\bibitem{CEGL2012}
Melisachew~Wudage Chekol, J{\'{e}}r{\^{o}}me Euzenat, Pierre Genev{\`{e}}s, and
  Nabil Laya{\"{\i}}da.
\newblock {SPARQL} query containment under {SHI} axioms.
\newblock In {\em Proc.\ {AAAI} 2012}. {AAAI} Press, 2012.

\bibitem{CP2016}
Melisachew~Wudage Chekol and Giuseppe Pirr{\`{o}}.
\newblock Containment of expressive {SPARQL} navigational queries.
\newblock In {\em Proc.\ {ISWC} 2016}, volume 9981 of {\em Lecture Notes in
  Computer Science}, pages 86--101. Springer, 2016.

\bibitem{Chen14}
Hubie Chen.
\newblock The tractability frontier of graph-like first-order query sets.
\newblock In Thomas~A. Henzinger and Dale Miller, editors, {\em Joint Meeting
  of the Twenty-Third {EACSL} Annual Conference on Computer Science Logic
  {(CSL)} and the Twenty-Ninth Annual {ACM/IEEE} Symposium on Logic in Computer
  Science (LICS), {CSL-LICS} '14, Vienna, Austria, July 14 - 18, 2014}, pages
  31:1--31:9. {ACM}, 2014.
\newblock URL: \url{http://doi.acm.org/10.1145/2603088.2603119}, \href
  {http://dx.doi.org/10.1145/2603088.2603119}
  {\path{doi:10.1145/2603088.2603119}}.

\bibitem{CD2012}
Hubie Chen and V{\'{\i}}ctor Dalmau.
\newblock Decomposing quantified conjunctive (or disjunctive) formulas.
\newblock In {\em Proceedings of the 27th Annual {IEEE} Symposium on Logic in
  Computer Science, {LICS} 2012, Dubrovnik, Croatia, June 25-28, 2012}, pages
  205--214. {IEEE} Computer Society, 2012.
\newblock URL: \url{https://doi.org/10.1109/LICS.2012.31}, \href
  {http://dx.doi.org/10.1109/LICS.2012.31} {\path{doi:10.1109/LICS.2012.31}}.

\bibitem{ChenM13}
Hubie Chen and D{\'{a}}niel Marx.
\newblock Block-sorted quantified conjunctive queries.
\newblock In Fedor~V. Fomin, Rusins Freivalds, Marta~Z. Kwiatkowska, and David
  Peleg, editors, {\em Automata, Languages, and Programming - 40th
  International Colloquium, {ICALP} 2013, Riga, Latvia, July 8-12, 2013,
  Proceedings, Part {II}}, volume 7966 of {\em Lecture Notes in Computer
  Science}, pages 125--136. Springer, 2013.
\newblock URL: \url{https://doi.org/10.1007/978-3-642-39212-2_14}, \href
  {http://dx.doi.org/10.1007/978-3-642-39212-2_14}
  {\path{doi:10.1007/978-3-642-39212-2_14}}.

\bibitem{ChenGL17}
Yijia Chen, Martin Grohe, and Bingkai Lin.
\newblock The hardness of embedding grids and walls.
\newblock {\em CoRR}, abs/1703.06423, 2017.
\newblock URL: \url{http://arxiv.org/abs/1703.06423}.

\bibitem{DalmauKV2002}
V.~Dalmau, P.G. Kolaitis, and M.Y. Vardi.
\newblock {Constraint Satisfaction, Bounded Treewidth, and Finite-Variable
  Logics}.
\newblock In {\em International Conference on Principles and Practice of
  Constraint Programming 2002}, pages 310--326, 2002.

\bibitem{DurandM14b}
Arnaud Durand and Stefan Mengel.
\newblock Structural tractability of counting of solutions to conjunctive
  queries.
\newblock {\em Theory of Computing Systems}, pages 1--48, 2014.

\bibitem{FlumG06}
J{\"{o}}rg Flum and Martin Grohe.
\newblock {\em Parameterized Complexity Theory}.
\newblock Texts in Theoretical Computer Science. An {EATCS} Series. Springer,
  2006.
\newblock URL: \url{https://doi.org/10.1007/3-540-29953-X}, \href
  {http://dx.doi.org/10.1007/3-540-29953-X} {\path{doi:10.1007/3-540-29953-X}}.

\bibitem{GUKFC2016}
Floris Geerts, Thomas Unger, Grigoris Karvounarakis, Irini Fundulaki, and
  Vassilis Christophides.
\newblock Algebraic structures for capturing the provenance of {SPARQL}
  queries.
\newblock {\em J. {ACM}}, 63(1):7:1--7:63, 2016.

\bibitem{Grohe01}
Martin Grohe.
\newblock The parameterized complexity of database queries.
\newblock In Peter Buneman, editor, {\em Proceedings of the Twentieth {ACM}
  {SIGACT-SIGMOD-SIGART} Symposium on Principles of Database Systems, May
  21-23, 2001, Santa Barbara, California, {USA}}, pages 82--92. {ACM}, 2001.
\newblock URL: \url{http://doi.acm.org/10.1145/375551.375564}, \href
  {http://dx.doi.org/10.1145/375551.375564} {\path{doi:10.1145/375551.375564}}.

\bibitem{Grohe02}
Martin Grohe.
\newblock Parameterized complexity for the database theorist.
\newblock {\em {SIGMOD} Record}, 31(4):86--96, 2002.
\newblock URL: \url{http://doi.acm.org/10.1145/637411.637428}, \href
  {http://dx.doi.org/10.1145/637411.637428} {\path{doi:10.1145/637411.637428}}.

\bibitem{Grohe2007}
Martin Grohe.
\newblock The complexity of homomorphism and constraint satisfaction problems
  seen from the other side.
\newblock {\em Journal of the ACM}, 54(1), 2007.

\bibitem{GroheSS01}
Martin Grohe, Thomas Schwentick, and Luc Segoufin.
\newblock When is the evaluation of conjunctive queries tractable?
\newblock In Jeffrey~Scott Vitter, Paul~G. Spirakis, and Mihalis Yannakakis,
  editors, {\em Proceedings on 33rd Annual {ACM} Symposium on Theory of
  Computing, July 6-8, 2001, Heraklion, Crete, Greece}, pages 657--666. {ACM},
  2001.
\newblock URL: \url{http://doi.acm.org/10.1145/380752.380867}, \href
  {http://dx.doi.org/10.1145/380752.380867} {\path{doi:10.1145/380752.380867}}.

\bibitem{KK2016}
Mark Kaminski and Egor~V. Kostylev.
\newblock Beyond well-designed {SPARQL}.
\newblock In {\em Proc.\ {ICDT} 2016}, volume~48 of {\em LIPIcs}, pages
  5:1--5:18. Schloss Dagstuhl - Leibniz-Zentrum fuer Informatik, 2016.

\bibitem{KKG2017}
Mark Kaminski, Egor~V. Kostylev, and Bernardo~Cuenca Grau.
\newblock Query nesting, assignment, and aggregation in {SPARQL} 1.1.
\newblock {\em {ACM} Trans. Database Syst.}, 42(3):17:1--17:46, 2017.

\bibitem{KRRV2015}
Egor~V. Kostylev, Juan~L. Reutter, Miguel Romero, and Domagoj Vrgoc.
\newblock {SPARQL} with property paths.
\newblock In {\em Porc.\ {ISWC} 2015}, volume 9366 of {\em Lecture Notes in
  Computer Science}, pages 3--18. Springer, 2015.

\bibitem{KRU2015}
Egor~V. Kostylev, Juan~L. Reutter, and Mart{\'{\i}}n Ugarte.
\newblock {CONSTRUCT} queries in {SPARQL}.
\newblock In {\em Proc.\ {ICDT} 2015}, volume~31 of {\em LIPIcs}, pages
  212--229. Schloss Dagstuhl - Leibniz-Zentrum fuer Informatik, 2015.

\bibitem{KPS2016}
Markus Kr\"oll, Reinhard Pichler, and Sebastian Skritek.
\newblock On the complexity of enumerating the answers to well-designed pattern
  trees.
\newblock In {\em Proc.\ {ICDT} 2016}, volume~48 of {\em LIPIcs}, pages
  22:1--22:18. Schloss Dagstuhl - Leibniz-Zentrum fuer Informatik, 2016.

\bibitem{LPPS2013}
Andr{\'{e}}s Letelier, Jorge P{\'{e}}rez, Reinhard Pichler, and Sebastian
  Skritek.
\newblock Static analysis and optimization of semantic web queries.
\newblock {\em {ACM} Trans. Database Syst.}, 38(4):25, 2013.

\bibitem{Lin15}
Bingkai Lin.
\newblock The parameterized complexity of \emph{k}-biclique.
\newblock In Piotr Indyk, editor, {\em Proceedings of the Twenty-Sixth Annual
  {ACM-SIAM} Symposium on Discrete Algorithms, {SODA} 2015, San Diego, CA, USA,
  January 4-6, 2015}, pages 605--615. {SIAM}, 2015.
\newblock URL: \url{https://doi.org/10.1137/1.9781611973730.41}, \href
  {http://dx.doi.org/10.1137/1.9781611973730.41}
  {\path{doi:10.1137/1.9781611973730.41}}.

\bibitem{LM2013}
Katja Losemann and Wim Martens.
\newblock The complexity of regular expressions and property paths in {SPARQL}.
\newblock {\em {ACM} Trans. Database Syst.}, 38(4):24:1--24:39, 2013.

\bibitem{Marx10}
D{\'{a}}niel Marx.
\newblock Tractable hypergraph properties for constraint satisfaction and
  conjunctive queries.
\newblock In Leonard~J. Schulman, editor, {\em Proceedings of the 42nd {ACM}
  Symposium on Theory of Computing, {STOC} 2010, Cambridge, Massachusetts, USA,
  5-8 June 2010}, pages 735--744. {ACM}, 2010.
\newblock URL: \url{http://doi.acm.org/10.1145/1806689.1806790}, \href
  {http://dx.doi.org/10.1145/1806689.1806790}
  {\path{doi:10.1145/1806689.1806790}}.

\bibitem{PapadimitriouY99}
Christos~H. Papadimitriou and Mihalis Yannakakis.
\newblock On the complexity of database queries.
\newblock {\em J. Comput. Syst. Sci.}, 58(3):407--427, 1999.
\newblock URL: \url{https://doi.org/10.1006/jcss.1999.1626}, \href
  {http://dx.doi.org/10.1006/jcss.1999.1626}
  {\path{doi:10.1006/jcss.1999.1626}}.

\bibitem{PAG2006}
Jorge P{\'{e}}rez, Marcelo Arenas, and Claudio Gutierrez.
\newblock Semantics and complexity of {SPARQL}.
\newblock In {\em Proc.\ {ISWC} 2006}, volume 4273 of {\em LNCS}, pages 30--43.
  Springer, 2006.

\bibitem{PAG2009}
Jorge P{\'{e}}rez, Marcelo Arenas, and Claudio Gutierrez.
\newblock Semantics and complexity of {SPARQL}.
\newblock {\em {ACM} Trans. Database Syst.}, 34(3), 2009.

\bibitem{PV2011}
Fran{\c{c}}ois Picalausa and Stijn Vansummeren.
\newblock What are real {SPARQL} queries like?
\newblock In {\em Proc.\ {SWIM} 2011}, page~7. {ACM}, 2011.

\bibitem{PS2014}
Reinhard Pichler and Sebastian Skritek.
\newblock Containment and equivalence of well-designed {SPARQL}.
\newblock In {\em Proc.\ {PODS} 2014}, pages 39--50. {ACM}, 2014.

\bibitem{Pol2007}
Axel Polleres.
\newblock From {SPARQL} to rules (and back).
\newblock In {\em Proc.\ {WWW} 2007}, pages 787--796. {ACM}, 2007.

\bibitem{PW2013}
Axel Polleres and Johannes~Peter Wallner.
\newblock On the relation between {SPARQL1.1} and answer set programming.
\newblock {\em Journal of Applied Non-Classical Logics}, 23(1-2):159--212,
  2013.

\bibitem{SML2010}
Michael Schmidt, Michael Meier, and Georg Lausen.
\newblock Foundations of {SPARQL} query optimization.
\newblock In {\em Proc.\ ICDT 2010}, pages 4--33. {ACM}, 2010.

\bibitem{ZB2014}
Xiaowang Zhang and Jan~Van den Bussche.
\newblock On the primitivity of operators in {SPARQL}.
\newblock {\em Inf. Process. Lett.}, 114(9):480--485, 2014.

\bibitem{ZBP2016}
Xiaowang Zhang, Jan~Van den Bussche, and Fran{\c{c}}ois Picalausa.
\newblock On the satisfiability problem for {SPARQL} patterns.
\newblock {\em J. Artif. Intell. Res.}, 56:403--428, 2016.

\end{thebibliography}

\end{document}